\documentclass{amsart}
\usepackage{amssymb,mathrsfs}%MnSymbol}%,skak}
\usepackage{MnSymbol}

\usepackage[textheight=25 cm, textwidth=16cm, headheight=10pt, headsep=20pt, footskip=10pt]{geometry}
\usepackage{enumitem}

\newcommand{\cal}{\mathcal }

\newcommand{\R}{{\mathbb R}}

\renewcommand{\a}{\alpha}

\newcommand{\pa}{\partial}

\def\be{\begin{equation}}
\def\ee{\end{equation}}
\def\bea{\begin{eqnarray}}
\def\eea{\end{eqnarray}}

\def\nn{\nonumber}

\def\o{\omega}

\def\bC {\mathbf{C}}

\def\bR {\mathbf{R}}

\def\fH {\mathfrak{H}}

\def\cD {\mathcal{D}}
\def\cE {\mathcal{E}}

\def\cL {\mathcal{L}}

\def\a {{\alpha}}

\def\rstr {{\big |}}

\newcommand{\Tr}{\operatorname{trace}}

\newcommand{\op}[1]{\overline{\underline{#1}}}

\newcommand{\ba}{\begin{aligned}}
\newcommand{\ea}{\end{aligned}}

\newtheorem{Thm}{Theorem}[section]
\newtheorem{Rmk}[Thm]{Remark}
\newtheorem{Prop}[Thm]{Proposition}
\newtheorem{Cor}[Thm]{Corollary}
\newtheorem{Lem}[Thm]{Lemma}
\newtheorem{Def}[Thm]{Definition}

\usepackage{color}
\newcommand{\norm}[1]{\lVert #1 \rVert}
\newcommand{\marg}[1]{\mbox{Tr}^{#1}}

\newcommand{\margtracej}[3]{{#2}^{\otimes \{#1\}}_{#3}}

\newcommand{\margetracej}[3]{{[#2]}^{\otimes {#1}}_{#3}}
\newcommand{\lone}[1]{\lon^{\otimes #1}}
\newcommand{\lono}[1]{\lon^{\otimes #1}}
\newcommand{\lon}{\mathbb L}

\newcommand{\tr}{\mbox{Tr}}
\newcommand{\FF}{F}
\newcommand{\unN}{1,\dots, N}
\newcommand{\un}[1]{1,\dots, #1}
\newcommand{\suml}{\sum\limits}

\newcommand{\sumls}[1]{\suml_{\substack{#1}}}

\newcommand{\bcr}{\begin{color}{red}}
\newcommand{\bcb}{\begin{color}{blue}}
\newcommand{\ec}{\end{color}}

\newcommand{\ksq}{For the {\bf K}, {\bf S} and {\bf Q} models,\ }
\newcommand{\Pii}[2]{\underset{#1}{\overset{#2}{\Pi}}}

\begin{document}

\title[Asymptotic expansion of the mean-field approximation]{\LARGE Asymptotic expansion of the mean-field approximation}

\author[T. Paul]{\Large Thierry Paul}
\address[T.P.]{CMLS, Ecole polytechnique, CNRS, Universit\'e Paris-Saclay, 91128 Palaiseau Cedex, France}
\email{thierry.paul@polytechnique.edu}

\author[M. Pulvirenti]{\Large Mario Pulvirenti}
\address[M.P.]{Sapienza Universit\`a di Roma, Dipartimento di Matematica, Piazzale Aldo Moro 5, 00185 Roma}
\email{pulvirenti@mat.uniroma1.it} 

%\author[S. Simonella]{Sergio Simonella}
%\address[S.S.]{ Zentrum Mathematik, TU M\"{u}nchen \\ Boltzmannstrasse 3, 85748 Garching -- Germany}
%\email{s.simonella@tum.de}
%\author{M.P $\&$ T.P.}

\maketitle
\LARGE
\begin{abstract}
We established and estimate the full asymptotic expansion in integer powers of $\frac1N$ of the $[\sqrt N]$ first marginals of  $N$-body systems evolution lying in a general paradigm containing Kac models and non-relativistic quantum evolution. We prove that the coefficients of the expansion are, at any time, explicitly computable given the knowledge of the linearization on the  one-body  meanfield kinetic limit equation. Instead of working directly with the corresponding  BBGKY-type hierarchy, we follows a method developed in \cite{PPS} for the meanfield limit, dealing with error terms analogue to the $v$-functions used in previous works. As a by-product we get that the rate of convergence to the meanfield limit in $\frac1N$ is optimal.
\end{abstract}

%\tableofcontents
\LARGE

\section{Introduction: motivation  and main results}\label{into}
Mean field limit concerns systems of interacting (classical or quantum) particles whose  number diverges in a way linked with a rescaling of the interaction insuring an equilibrium between  interaction  and residual kinetic energies. In the case of an additive one-body kinetic energy part and a two-body interaction, and without taking in consideration quantum statistics, this equilibrium is reached by putting in front of the interaction a coupling constant proportional to the inverse of the number of particles.

The system is then described by isolating the evolution of one (or $j$) particle(s) and averaging over all the other. This leads to a partial information on the system driven by the so-called $j$-marginals. The mean field theory insures 
%then
%the fact 
that the $j$-marginals tend, as the number of particles diverges, to the $j$-tensor power of the solution of a non-linear one-body meanfield equation (Vlasov, Hartree,...) issued from the $1$-marginal on the initial $N$-body state. This program has be achieved in many different situations, and the literature concerning the mean field approach is protuberant. We refer to the review article \cite{SSS} for a reasonable bibliography.

Much less is known about the fluctuations around this limit, namely the correction to be added  to the factorized limit  in order to get better approximations of the true evolution of the $j$-marginals. 

The identification of the leading order of these fluctuations with a Gaussian stochastic process has been established in the quantum context first in \cite{Hepplieb} and in the classical one in \cite{BraunHepp}. For the classical dynamics of hard spheres, the fluctuations around the Boltzmann equation have been computed at leading order in \cite{Spohnboltz}, generalizing to non-equilibrium states the results of \cite{LandSpohn}. More recently, for the quantum case, fluctuations near the Hartree dynamics  has been derived in \cite{LNS} ( after \cite{LNSS}) and in \cite{BPS} also for the grand canonical ensamble formalism (number of particles non fixed),  using in both cases the methods of second quantization (Fock space).
\vskip 1cm
Recently, we developed (together with S. Simonella) in \cite{PPS}  a method to derive mean field limit, alternative to the ones using empirical measures or direct estimates on the  ``BBGKY-type" hierarchies
(systems of coupled  equations satisfied by the set of the $j$-marginals).
This method
%, 
%inspired by the theory of ``U-functions" in statistical mechanics and in particular in \cite{PS}, 
%These quantities have been already used (under the name ``$v$-functions'') to deal with kinetic limits of stochastic models
%\cite{ DP, CDPP91,  BDP,  DOPT,  DOPT1, CP, CPW,  DPTV} and they have been recently investigated in the more singular low density limit of hard spheres \cite{PS}.
rather uses the hierarchy followed by the ``kinetic errors" $E_{j-k}$ (defined below),
 already used (under the name ``$v$-functions'') to deal with kinetic limits of stochastic models
\cite{ DP, CDPP91,  BDP,  DOPT,  DOPT1, CP, CPW,  DPTV} and  recently investigated in the more singular low density limit of hard spheres \cite{PS}.
These quantities are, roughly speaking, the coefficient of the  decomposition of the $j$-marginal as a linear combination of  the $k$th tensor powers, $k=1,\dots,j$, of the solution of the mean-field equation issued from of the  $1$-marginal of the initial full state.
%\footnote{Note that, by the meanfield limit result, this decomposition is also the one in terms of leading orders of the marginals of lower orders.}.  
We developed in \cite{PPS} a strategy suitable in particular for Kac models (homogeneous original one \cite{Kac,Kac2} and non-homogeneous \cite{Cer}) and quantum mean field theory. This strategy allowed us to derive the limiting factorization property of the $j$-marginals up to, roughly speaking, $j\lesssim \sqrt N$. This threshold is, on the other side, the  one obtained by heuristic arguments as shown in \cite{PPS}. 
 
 Here and in all this article, $N$ denotes the number of particles of the system under consideration.
\vskip 0.5cm

In the present note we  provide and estimate a full asymptotic expansion in powers  of $\frac1N$ of the difference between the evolution of $j$-marginals and its factorized leading order form (Corollary \ref{cormain}), following a similar result for the kinetic errors $E_j(t)$ (Theorem \ref{main}). Our results are valid for $j\lesssim \sqrt N$ in  an abstract paradigm,  generalization of the   abstract formalism developed in \cite{PPS} and
%(see Remark ... in \cite{PPS}) 
 described in Appendix \ref{abstract}, and applies of course to the different  Kac models  and quantum mean field theory treated in \cite{PPS}.
Moreover, we show that the additional knowledge of the linearization of the mean field flow,   around  the meanfield  solution issued from the $1$-marginal of the initial data, gives an explicit construction of the full asymptotic expansion of the  $j$-marginals in powers of $\frac1N$.

%In the present article, we will establish the results just described in the following situations:
%\begin{itemize}
%\item the quantum mean-field model
%\item the homogeneous and non-homogeneous Kac models
%\item a general ``abstract" model containing the two first situations
%\end{itemize}
We will state in this section the quantum results and postpone in Section \ref{kac} and in the Appendix \ref{abstract} the Kac's type and the abstract results, respectively. Sections \ref{recur} and \ref{esti} contain the algebraic and anlytical proofs of our results in the quantum case, immediately transposable to the Kac and abstract situations as shown in Section \ref{kac} and Appendix \ref{abstract}. In Section \ref{explicit} one compute explicitly the first terms of the asymptotic expansions obtained in the quantum case and rely them to previous works.
\vskip 1cm
\subsection{Quantum mean-field}\label{quantum}

Let $\cL^1(L^2(\bR^d)$ be the space of trace class operators on $L^2(\bR^d)$, with their  associated norms.

We consider the evolution of a system of $N$ quantum particles interacting through a (real-valued) two-body, even potential $V$, described for any value of the Planck constant $\hbar>0$ by the Schr\"odinger equation 
$$
i\hbar\partial_t\psi=H_N\psi\,,\qquad\psi\rstr_{t=0}=\psi_{in}\in\fH_N:=L^2(\bR^d)^{\otimes N}\,,
$$
where 
$$
H_N:=-\tfrac12\hbar^2\sum_{k=1}^N\Delta_{x_k}+\frac1{2N}\sum_{1\leq k,l\leq N}V(x_k-x_l).
$$
We will suppose in the whole presnt paper that  the $N$-body Hamiltonian $H_N$  is essentially self-adjoint. 
%With the notation
%$$
%X_N:=(x_1,\ldots,x_N)\in(\bR^d)^N\,,
%$$
%the $N$-body Hamiltonian is recast as
%$$
%H_N={-}\tfrac12\hbar^2\Dlt_{X_N}+V(X_N)\,,
%$$
%where $V$ is the $N$-body potential in the mean-field scaling, i.e. with $1/N$ coupling constant:
%$$
%V(X_N):=\frac1{2N}\sum_{1\leq k,l\leq N}\Phi(x_k-x_l)\,.
%$$

Instead of the Schr\"odinger equation written in terms of wave functions, we shall rather consider the quantum evolution of density matrices. An $N$-body density matrix is an operator $F^N$ such that
$$
0\le F^N=(F^N)^*,\quad\Tr_{\fH_N}(F^N)=1\,.
$$
%(We denote by $\cL(\fH_N)$ the set of  bounded linear operators on $\fH_N$.)

The evolution of the density matrix $F^N\mapsto F^N(t)$ of a $N$-particle system is governed for any value of the Planck constant $\hbar>0$ by the von Neumann equation
 \be\label{Ns}
\partial_tF^N=\frac1{i\hbar}[H_N,F^N],
\ee equivalent to the Schr\"odinger equation when $F^N(0)$ is a rank one projector, modulo a global phase.

\noindent Positivity, norm  and trace are obviously preserved by \eqref{Ns} since $H_N$ is taken essentially self-adjoint.
%The density matrix $D_N$ for a $N$-particle system described in terms of the $N$-particle wave function $\Psi_N\equiv\Psi_N(t,\un x_N)$ is the orthogonal projection in $\fH_N$ on the one-dimensional subspace $\bC\Psi_N$. In other 
%words, $D_N(t)$ is the integral operator on $\fH_N$ with integral kernel $\Psi_N(t,\un x_N)\overline{\Psi_N(t,\un y_N)}$. Up to multiplying the $N$-body wave function by a global phase factor, both formulations of the quantum dynamics 
%of an $N$-particle system are equivalent.
%
%If $D_{in}^N$ is factorized, that is of the form 
%$$
%D^N_{in}=D_{in}^{\otimes N}
%$$
%where $D_{in}$ is a density matrix on $\fH:=L^2(\bR^d)$, then $D^N(t)$ is in general \textit{not }factorized for $t>0$. However, it is known that $D_N(t)$ has a ``tendency to become factorized'' for all $t>0$ as $N\to\infty$, i.e. in the 
%mean-field limit, \textit{for each} $\hbar>0$. 
%
%The precise formulation of the ``tendency to become factorized'' involves the notion of marginal of a density operator. 

For each $j=1,\ldots,N$, the $j$-particle marginal $F^N_j(t)$ of $F^N(t)$ is the unique operator on $\fH_j$ such that
$$
\Tr_{\fH_N}[F^N(t)(A_1\otimes\dots\otimes A_j\otimes I_{\fH_{N-j}})]=\Tr_{\fH_j}[F^N_j(t)(A_1\otimes\dots\otimes  A_j)]\,.
$$
for all $A_1,\ldots,A_j$ bounded operators on $\fH$.
 Alternatively and equivalently,  the $F^N_j$ can be defined by the  partial trace 
 %of the integral kernel 
 of $F^N$ on the $N-j$ last ``particles": defining $F^N$ through its integral kernel $F^N(
 %x_1,\dots,x_n;x'_1,\dots,x'_n
x_1,x'_1;\dots;x_N,x'_N 
 )$, the integral kernel of $F^N_j$ is defined as (see \cite{BGM})
 \bea\label{intmarg}
 F^N_j(
 %x_1,\dots,x_j;x'_1,\dots,x'_j
 x_1,x'_1;\dots;x_j,x'_j
 )
 :=
 (\tr^{j+1}\dots\tr^NF^N)(x_1,x'_1;\dots;x_j,x'_j)&\nn\\
 :=\int_{\bR^{d(N_j)}} F^N(x_1,x'_1;\dots;x_{j};x'_{j};x_{j+1},x_{j+1};\dots;x_N,x_N)dx_{j+1}dx_N.&
 \eea

It will be convenient for the sequel to rewrite \eqref{Ns} in the following operator form
\be\label{Nsop}
\partial_t F^N=(K^N+V^N)F^N
\ee
where $K^N,V^N$ are operators on $\cL^1(L^2(\bR^{Nd}))$ defined by
\be\label{knvn}
K^N=\frac1{i\hbar}[-\frac{\hbar^2}2\Delta_{\bR^{dN}},\cdot],\ \ \ V^N
%=\frac1{i\hbar}[V,\cdot]:
=\frac1{2N}\sum_{k,l}V_{k,l}\mbox{ with }V_{k,l}:=\frac1{i\hbar}[V(x_k-x_l),\cdot].
\ee
As mentioned before   the (essantial) self-adjointness of $H_N$ implies that
\be\label{normpropag}
\norm{e^{t(K^N+V^N)}}_{\cL^1(L^2(\bR^{d}))\to\cL^1(L^2(\bR^{d}))}=
\norm{e^{tK^N}}_{\cL^1(L^2(\bR^{d}))\to\cL^1(L^2(\bR^{d}))}=1,\ t\in\bR.\ee
We will  denote 
\be\label{bl}
\lon:=\cL^1(L^2(\bR^{d}))\mbox{ so that }\lone{n}=
\cL^1(L^2(\bR^{nd})),\ n=\unN,
\ee
and, with a slight abuse of notation, 
\be\label{norms}
\left\{
\begin{array}{l}
\norm{\cdot}_1 \mbox{ the trace norm on any }\lone{j},\\
\norm{\cdot} \  \mbox{ the operator norm on any }\cL(\lone{i},\lone{j})
%\mbox{ or }\lone{j+1}\to\lone{j}.
\end{array}\right.
\ee
for  $i,j=\unN$ (here $\cL(\lone{i},\lone{j})$ is the set of bounded operators form $\lone{i}$ to $\lone{j}$).

\medskip

%For $j=0,\dots,N$, the $j$-particle marginal of $F^N \in (\lone{N})^+_1$ is defined as the 
% the partial trace of order $N-j$ of $F^N $, that is
%\be
%\label
%{jmarg}
%F^N_j =\tr^N \tr^{N-1} \cdots \tr^{j+1} F^N, F_N^N:=F^N.
%\ee
A density matrix $F^n\in\lone{n}$ is called symmetric  if its integral kernel $F^n(x_1,x'_1;\dots;x_n,x'_n)$ is invariant by any permutation $$(x_i,x'_i)\leftrightarrow(x_j,x'_j),\ i,j=\un{n}.$$
Note that the symmetry of $F^N$ is preserved by the equation \eqref{Ns} due to the particular form of the potential.

We define, for $n=\unN$,
\be
\label{coneschro}
%(\lone{n})^+_1
\cD_n= \{ F \in \lone{n} \,\,|\,\, F >0, \quad \| F\|_1=1 \quad \text {and $F$ is symmetric}\}.
\ee
Note that $F^N_j\in \lone{j}$ ($F^N_0=1\in\lone{0}:=\bC$) and $F^N_j >0, \norm{F^N_j}_1=\norm{F^N}_1$, and 
obviously $\FF^N_j$ is symmetric as $\FF^N$. That is to say: $$F^N_j\in\cD_j.$$

\bigskip
%\subsubsection*{Hierarchies}

%Applying successively $\marg{N},\marg{N-1}\dots$ to \eqref{Neq} and using \eqref{trackf}, we get 

%We proved in \cite{PPS} that t
The family of $j$-marginals, $j=\unN$, are solutions of the BBGKY hierarchy of equations (see \cite{BGM})
%\be\label{bbgky}
%\FF
%For a sequence of marginals $\{ f^N_J \}$ we consider the evolution:
\be\label{eqhiera}
\pa_t \FF^N_j = \left(K^j+\frac {T_j}{N}\right) \FF^N_j+ \frac {(N-j)} N C_{j+1} \FF^N_{j+1}
\ee
%where $\a(j,N)=\frac {(N-j)} N $.
where:
\be
K^j=\frac1{i\hbar}[-\frac{\hbar^2}2\Delta_{\bR^{jd}},\cdot]
\ee
\be
T_j=\sum_{1\leq i < r \leq j} T_{i,r}\ \ \ \ \ \ \mbox{ with }\  T_{i,r}=V_{ir}
\ee
and
\be
\label{C}
C_{j+1}\FF^N_{j+1}
%=\marg{j+1}\left(\sum_{i\leq j}V_{i,j+1}\FF^N_{j+1}\right)
=\sum_{i=1} ^jC_{i,j+1}\FF^N_{j+1}
\ee
with
\bea
C_{i,j}:& \lon^{\otimes (j+1)} & \to\ \ \  \ \lon^{\otimes j}\nn\\
& C_{i,j+1}\FF^N_{j+1}\ &=\ \ \ \ \marg{j+1}\left(V_{i,j+1}\FF^N_{j+1}\right)\;,
\label{eq:Ci}
\eea
where $\marg{j+1}$ is the partial trace with respect to the $(j+1)$th variable, as in \eqref{intmarg}.
\vskip 1cm
%\noindent For the sequel, we will need a (bit) more sophisticated definition of marginals and other objects. Though the definitions might look a bit heavy, let us point out from the beginning that the proof will use very few fact concerning them. 

Note that,  for all $ i \leq j=\unN$,
%Thanks to the different compatibility relations expressed earlier and the definition \eqref{defV}, one immediately get that, $\forall i.j=\unN$,
\be\label{normtc}
\norm{T_j}\leq j^2\frac{\norm{V}_{L^\infty}}\hbar,
%\leq j\norm{V} \mbox{ and }
\mbox{ and }\norm{C_{i,j+1}}\leq j\frac{\norm{V}_{L^\infty}}\hbar.
\ee
(meant for $\norm{T_j}_{\lon^{\otimes j}\to\lon^{\otimes j}}\mbox{ and } \norm{C_{i,j+1}}_{\lon^{\otimes (j+1)}\to\lon^{\otimes j}}
%,\ \norm{V}_{\lon^{\otimes 2}\to\lon^{\otimes 2}}
$ in accordance with \eqref{norms}).

\medskip
\vskip 1cm
%{ We introduce the non-linear mapping $Q(\FF,\FF)$, $Q: \lon \times \lon \to \lon$ by the formula
%\be\label{nlc}
%Q(\FF,\FF)=\marg{2}(V_{1,2}(\FF\otimes\FF))
%\ee
The Hartree equation is
\be\label{mfeh}
i\hbar\pa_t\FF=[-\frac{\hbar^2}2\Delta+V_\FF(x),\FF],\ \ 
%0\leq F(0)\in\cL^1(L^2(\bR^d))
%%,\ F(0)\geq 0
%,\ \norm{F(0)}_1=1.
F(0)\in\cD_1,
\ee
where $V_F(x)=\int_{\bR^d}V(x-y)F(y,y)dy$, $F(y,y')$ being the integral kernel of ~$F$.

Note that \eqref{mfeh} reads also 
\be\label{mmmmf}
\partial_tF=K^1F+Q(F,F),
\ee
 with 
 \be\label{deffffq}
 Q(\FF,\FF)=\tr^2(V_{1,2}(\FF\otimes\FF)).
 \ee

Since $V$ is bounded, \eqref{mfeh} has  for all time a unique solution $F(t)>0$ and $\norm{F(t)}
%=\norm{F(0)}
=1$ (see again \cite{BGM}).

\vskip 1cm

%\subsubsection*{Correlation error.}
In order to define the correlation error in an easy way, we need a bit more of notations concerning the variables of integral kernels.

 For $i\leq j=1,\dots,N,$ we  define the variables $z_i=(x_i,x'_i)$, and 
$Z_j=(z_1,\dots,z_j)$. For  $\{i_1,\cdots,i_k\}\subset \{1,\cdots,j\}$, we  denote
by $Z_{j}^{/\{i_1,\cdots,i_k\}}\in\bR^{2(j-k)d}$,  the vector $Z_j:=(z_1,\dots,z_j)$ after removing the components $z_{i_1},\dots z_{i_k}$.

\begin{Def}\label{defej}
For any 
$j=\unN$,
we define the  
 %{\color{blue} ``correlation error coefficients" } 
%{ ``kinetic error " }
{\em correlation error} $E_j\in~\lone{j}$ by its integral kernel
\be\label{deferrorq}
E_j(Z_j)=
\sum_{k=0}^j\sum_{1\leq i_1<\dots<i_k\leq j}(-1)^k
\FF(z_{i_1})\dots F(z_{i_k})F^N_{j-k}(Z_{j}^{/\{i_1,\cdots,i_k\}}).
\ee
By convention and consistently with $F^N_0=\norm{F}=1$, we define 
\be\label{ezero}
E_0:=1\in\lone{0}:=\bC.
\ee.
\end{Def}

In \cite{PPS} was shown that \eqref{deferrorq} is inverted by the following equality.
\be\label{invdeferrorq}
\FF^N_j(Z_j)=\sum_{k=0}^j\sum_{1\leq i_1<\dots<i_k\leq j}
\FF(z_{i_1})\dots F(z_{i_k})E_{j-k}(Z_{j}^{/\{i_1,\cdots,i_k\}}),\ j=0,\dots,N..
\ee
i.e. $F^N_j$ is the operator of integral kernel given by \eqref{invdeferrorq}.

% Note that, in order  to avoid too heavy notations, we made, in \eqref{deferrorq} and \eqref{invdeferrorq}, a slight abuse of notation as writing integral kernels as functions of $Z_j=((x_1,x'_1),\dots,(x_j,x'_j))$ 
% and not functions of $(x_1,\dots,x_j;x'_1,\dots,x'_j)$.
%\vskip 1cm 
% In the sequel, we will take $\hbar=1$, the correct value being easy to restore in the different expressions.
\vskip 1cm

\subsection{Main results of \cite{PPS}}\label{resultsPPS}
 Theorem 2.4,  theorem 2.1 and Corollary 2.2 in \cite{PPS} state the following facts.
 
 The kinetic errors $ E_j,\ j=\unN,$
 satisfy the system of equations
\bea\label{eqhieraerror111}
\pa_t E_j&=& \left(K^j+\frac 1{N}T_j\right) E_j +
%E_{J}.
D_j
E_{{j}} \nn \\
&+&  
D_j^1
E_{j+1} + 
D^{-1}_j
E_{j-1} +
D_j^{-2}
E_{j-2},
\eea
where the  operators $D_J,D_j^1,D^{-1}_j,D_j^{-2},\ j=0,\dots, N$, are defined at the beginning of the Section \ref{recur}, formulas \eqref{ddefj}-\eqref{defd=}.

\vskip 1cm
\begin{Thm}[Theorem 2.2. and Corollary 2.3 in \cite{PPS}]\label{mainpps}
\ 

\noindent Let $E_j(0)$ satisfy  
for some $C_0>1,B>0$
\be\label{condint}
\left\{
\begin{array}{rcl}
\norm{E_1(0)}_1&\leq & \frac BN\\
\norm{E_j(0)}_1&\leq &(\frac{j^2}N)^{j/2}C_0^j,\  j\geq 2.
\end{array}
\right.
\ee
 \label{pagefoot1}
 Then,
 for all  $t>0$ and all $j=~1,\dots, N$, 
 one has 

\be\label{timet}
\left\{
\begin{array}{rcl}
\norm{E_1(t)}_1&\leq &\frac 1N \big( B_2 e^{\frac{B_1 t \|V\|_{L^\infty}}\hbar } \big)\;
\\
\| E_j(t) \|_1 &\leq &\left( C_2 e^{\frac{C_1 t\norm{V}_{L^\infty}}\hbar} \right)^j \left(\frac{j}{\sqrt N}\right)^j ,\  j\geq 2.
\end{array}
\right.
\ee
for some  $B_1,C_1>0,\ B_2,C_2\geq 1 $  explicit (see 
%formulas (78),(81) 
Theorem 2.2
in \cite{PPS}), 

\noindent and
\be\label{eqcormain}
\| \FF^N_j (t) - \FF(t)^{\otimes j}   \|_1 \leq 
%\| E_1 (t) \|_1+ 8 \frac { (ejC_0)^2 e^{2C_1t\norm{V}}}N
 D_2 e^{\frac{D_1t\norm{V}_{L^\infty}}\hbar}\frac {j^2 }N,
 \ee
where $D_2=\sup\{B_2,(eC_0)^2\},\  B_1=\sup\{B_1,2C_1\}$.
\end{Thm}
%\begin{Thm}\label{main}
% Let us suppose  
%that, for all $j=1,\dots, N $ and for some $C_0\geq 1$,
%\be
%\label{hyp}
%\norm{E_j(0)}_1\leq C_0^j\left(\frac{j}{\sqrt N}\right)^{j}.
%\ee
%Then,
% for all  $t>0$ and all $j=~1,\dots, N$, 
% one has 
%\be
%\label{eqmain}
%\| E_j(t) \|_1 \leq \left( C_2 e^{C_1 t\norm{V}} \right)^j \left(\frac{j}{\sqrt N}\right)^j 
%\ee
%where $C_1>0,\ C_2\geq 1 $ are explicitly computable (formula \eqref{c1c2} below). 
%
%If we suppose in addition that $\norm{E_1(0)}_1\leq \frac {B_0}N$ for some $B_0>0$. Then for all $t > 0$
%\be\nonumber
%\norm{E_1(t)}_1\leq \frac 1N \big( B_2 e^{B_1 t \|V\| } \big)\;
%\ee
%for some  $B_1>0,\ B_2\geq 1 $  explicitly computable (formula \eqref{b1b2} below). 
%\end{Thm}
%
%%EQUIVALENCE.
%
%\begin{Cor}\label{cormain}
% For all  $t>0$ and all $j=~1,\dots, N$, one has
%\be\label{eqcormain}\nn
%\| \FF^N_j (t) - \FF(t)^{\otimes j}   \|_1 \leq 
%%\| E_1 (t) \|_1+ 8 \frac { (ejC_0)^2 e^{2C_1t\norm{V}}}N
% D_2 e^{D_1t\norm{V}}\frac {j^2 }N
%\ee
%where $D_2=\sup\{B_2,(eC_0)^2\},\  B_1=\sup\{B_1,2C_1\}$.
%\end{Cor}
\vskip 1.5cm
\subsection{Asymptotic expansion and main result of the present article}\label{mainresult}\ 

\noindent Two questions arise naturally:
\begin{enumerate}
\item are the estimates \eqref{timet} sharp?
\item  could \eqref{eqcormain} be improved with a r.h.s. of any order we wish?\footnote{Of course \eqref{timet} and \eqref{eqcormain} imply that
\be\label{grossier}
\norm{\FF^N_j (t) - \FF^{\otimes j}  
-\sumls{J\neq K {\subset} J\\|K|\leq k\leq|J|} 
\FF^{\otimes{|K|}}E^N_{|J|-|K|}}=O(N^{-(k+1)/2}),
\ee
but first one cannot go further in the approximation, and second \eqref{grossier} is meaningless without the knowledge of the $E^N_j$s.}
\end{enumerate}
%In view of \eqref{deferror} and \eqref{eqmain}, it is natural to wonder if \eqref{eqcormain} could be improved with a r.h.s. at any order we wish. 
\vskip 1cm

We will see below that, indeed, not only the estimates \eqref{timet} are true, but $N^{j/2}E_j^N(t)$ has a full asymptotic expansion in positive powers of $(\frac1N)^{\frac12}$ (actually we will show that this expansion contains only powers  $(\frac1N)^{\frac k2}$ with $k+j$ even)
%($\tfrac1N$ for $j=1$)
 if $N^{j/2}E_j^N(0)$ do posses such an expansion in half powers of $1/N$.
 
 More precisely we will show that, under the hypothesis \eqref{condint} on the initial data and  for all  time $t$ and all  $j=\unN$, 
 %there exist quantities  $\cE_j\sim\suml_{\ell=0}^\infty\cE_j^\ell N^{-\ell}$ such that
 \be\label{defce}
 E_j(t)=N^{-j/2}\cE_j,\ \ \ \cE_j(t)\sim\suml_{\ell=0}^\infty\cE_j^\ell(t) N^{-\ell/2} \mbox{ with }\cE^k_j(t)=0\mbox{ for } j+k\mbox{ odd}
 \ee 
 when the same is true  at $t=0$.
 
\vskip 1cm
Moreover we will see that all the $\cE_j^\ell$ can be explicitly recursively computed after the knowledge of the linearization of the mean field equation \eqref{mfeh} around the solution of \eqref{mfeh} with initial condition $\FF(0)=(\FF^N(0))_1$. Indeed the proof will involve the ``$j$-kinteic linear mean field flow" 
%$U_j(t,s)$ 
defined by the linear kinetic  mean field equation of order $j$:
\be\label{linhartree}
\frac d{dt}A(t)=(K^j+\Delta_j(t))A(t),\ \ \ A(0)\in\lone{j},
\ee
where $\Delta_j(t)=\lim\limits_{N\to\infty}D_j(t)
%=\tfrac N{N-j}D_j(t)
$. 

\noindent\eqref{linhartree} is
solved by the two parameter semigroup $U^9_J(t,s)$ solving
\bea\label{jkin}
\partial _t U_j^{0}(t,s)&=&(K^j
%+\tfrac{T_j}N
%\tfrac{N-j}NV_\FF^j
+\Delta_j(t))U_j^0(t,s).\\ 
U_j^0(s,s)&=&I.\nn
\eea
%Here $\Delta_j(t)=\lim\limits_{N\to\infty}D_j(t)
%%=\tfrac N{N-j}D_j(t)
%$ and 
Note that $U_j^0$ exists since $K^j$ generates a unitary flow and $\Delta_j$ is bounded. 

The reason of the terminology comes form the fact that, as shown by \eqref{ddefj},  
%$T_1=0,\ 
$\Delta_1=Q(\FF,\cdot)+Q(\cdot,F)$ so that, for $j=1$, \eqref{linhartree} is the linearization  of the mean field equation \eqref{mfeh} around its solution $\FF(t)$. Note moreover that, for $G^1,G^2\in\lon$,
%\be\label{deltaj}
%\Delta_=j(\tr^1+\tr^{j+1})( V_{1,j+1}\cdot\otimes\FF)
%\ee
%where 
%\be\label{vjf}
%V_\FF^j\varphi_j=\suml_{i=1}^j\marg{j+1}V_{1,j+1}\varphi_j\otimes\FF,
%\ee
\bea\label{delta2}
&&\Delta_2(G^1G^2+G^2G^1)=\\
%j(\tr^1+\tr^{j+1})( V_{1,j+1}\cdot\otimes\FF)
&&(\Delta_1 G^1) G^2+ G^1(\Delta_1 G^2)+(\Delta_1 G^2) G^1+ G^2 (\Delta_1 G^1).\nn
\eea
and therefore
\be\label{u2ui}
U^0_2(t,s)(G^1G^2+G^2G^1)=(U^0_1(t,s)G^1)(U^0_1(t,s)G^2).
\ee
More generally, if $P_j:\ \lon^j\to\lone{j}$ is any homogeneous polynomial invariant by permutations,
\be\label{uju1}
U^0_J(t,s)P_j(G^1,\dots,G^j)
=
P_j(U^0_1(t,s)G^1,\dots,U^0_1(t,s)G^j).
\ee
That is: $U^0_J$ drives each $G^j$ along the linearized mean field flow ``factor by factor". Denoting by 
$\lon^{\overset{}{\otimes} j}_{sym}$ the subspace of symmetric (by permutations) vectors, 
%in the sense of \eqref{defsym}, 
%of $\lono{j}$,
we just prove the following result.
\begin{Lem}\label{lemsym}
\be\nn
U_j^0(t,s)|_{\lon^{{\otimes} j}_{sym}}=U_1^0(t,s)^{\otimes j}.
\ee
\end{Lem}
\vskip 0.5cm
\noindent Note also that $U_j^0(t,s)$ is given by a convergent Dyson expansion and that,  by the isommetry of the flow generated by $K^j
%+\tfrac{T_j}N
$ and the bound \eqref{endowq} below, we have by Gronwall Lemma that $\norm{U_j^o(t,s)}\leq e^{j|t-s|}$. 

We will also need  in the sequel the semigroup defined by 
\bea\label{jkinj}
\partial _t U_j^{}(t,s)&=&(K^j
+\tfrac{T_j}N
%\tfrac{N-j}NV_\FF^j
+\tfrac{N-j}ND_j(t))U_j(t,s).\\ 
U_j(s,s)&=&I.\nn
\eea
%Note that $\frac{N-j}ND_j=\Delta_j-\frac jND_j$.

\noindent$U_j(t)$ exists  by the same argument as for $U_j^0(t)$. Moreover,  in the regime $j^2/N$ small, $U_j$ can be also computed out of $U_j^0$ by a convergent perturbation expansion (in $j^2/N$). Indeed \eqref{ddefj}, \eqref{normtc} and \eqref{endowq} show clearly that $\norm{\tfrac{T_j}N+\tfrac{N-j}ND_j-\Delta_j}\leq \tfrac{j^2}N$. Therefore, for any $t$, $U_j(t)$ can be approximated,   up to any power of $j^2/N$, by a finite Dyson expansion.

%We will also show that the sub-leading order in the expansion of $\FF_j$ is, somehow like in the case of the Boltzmann equation, driven by the \textit{recollision operator} defined by
%\be\nn
%R(\FF,\FF):=V\FF\otimes \FF,
%\ee
%(see Section \ref{spohn81}).
\vskip 1cm
Finally, we will extend \eqref{eqcormain} as we will show that $F^N_j$ has an asymptotic expansion in positive powers of $1/N$ whose partial sums up to any order $n\geq 0$ is $O(jN^{-n-1})$-close to $F^N_j$.

%any power in $(\tfrac1N)^{1/2}$: denoting $\cE_{j,k}:=\suml_{\ell=1}^k\cE_j^\ell N^{-\ell/2}$, we will 
%estimate the dependence in $j$ and $k$ of the r.h.s. of
%\be
%\norm{\FF^N_j-\FF^{\otimes j}-N^{-1}\FF^{j-1}\cE_{1,k}-\sum_{m=2}^{
%%\inf{(j,k)}
%j}
%N^{-\tfrac m2}\FF^{j-m}\cE_{m,k+2-m}}=O(N^{-(k+2)}).\nn
%\ee
\vskip 1cm
The main results of the present note are the following.
\begin{Thm}\label{main}
%Let $\cE_j^k(t),\ j=1,\dots, k=0,\dots,t>0,$ be defined recursively out of 
%$\cE_j^k(0),\ j=1,\dots, k=0,\dots,$ by 
%(see Proposition \ref{asympt} below)
Consider for $j=0,\dots,N,\  k=0,\dots,\ t\geq 0$ the system of recursive relations
\be\label{eqrecur}
\left\{
\begin{array}{rcl}
\cE_j^k(t)&=&U_j(t,0)\cE_j^k(0)\\
&+&\int_{s=0}^t
%\big[
U_j(t,s)(\Delta_j^=\cE_{j-2}^k(s)
 +\Delta_j^+\cE_{j+1}^{k-1}(s)+\Delta_j^-\cE_{j-1}^{k-1}(s))
% \big]
 ds\\
\cE_0^k(t)&=&\delta_{k,0},
\\
\cE_{-1}^k(t)&=&\cE_{-2}^k(t)\ =\ \cE_j^{-1}(t)\ =\ 0\mbox{ by convention.}
%\ \ \cE_{-1}^k=\cE_{-2}^k=\cE_j^{-1}:=0
\end{array}\right.
\ee
where $U_j(t,s)$ is 
%the
%% $j$-kinetic 
% semigroup 
 the two times flow defined by \eqref{jkinj} and  $\Delta^+_jD_j^1,\ \Delta^-_j=ND^{-1}_j,~\  \Delta^=_j=ND_j^{-2}$, the $D_j$s being given by \eqref{ddefj}-\eqref{defd=}.
 
Then, for any $j=1,\dots,N,\ k=0,\dots,\ t\geq 0,$ the knowledge of 
$\cE_{j'}^{k'}(0)$ for $j'=1,\dots,j+k,\  k'=0,\dots,k$, determine in a unique way  $\cE_j^k(t)$, and 
\be\nn
\cE^k_j(t)=0\mbox{ when }j+k\mbox{ is odd}
\ee
if $\cE^k_j(0)$ satisfies the same property.
%(in particular $\lim\limits_{N\to\infty}N^{m+\frac12}E_{2m+1}=0\ \mbox{ for all } m\in\bN$).

%Moreover, if $E_j(t):=N^{-j/2}\cE_j(t)$ solves the equation \eqref{eqhieraerror111} and satisfies \eqref{condint}, and $\cE_j(0)$ has an asymptotic expansion 
%%on the form 
%$\cE_j(0)\sim\sumls{k=0\\k+j\ even}^\infty N^{-k/2}\cE_j^k(0)$, one has
%$$
%\norm{\cE_j(t)-\sum_{k=0}^{2n}N^{-k/2}\cE_j^k(t)}_1
%\leq
%%tC_{2n}(t)
%D_{2n}(t)
%N^{-n-\frac12}(
%%e^{Ct\norm{V}}
%%A^{2n}_t
%D'_{2n}(t)
%j^2)^{j/2},
%$$
%%for any function $\sca$ as in footnote 1 page \pageref{pagefoot1},  and  
%where $D_k(t),D'_k(t)$ are defined in \eqref{dnt} below.

Moreover, for all $t\in\bR$,  the solution $E_j(t)$ of   \eqref{eqhieraerror111} with initial data $E_j(0)$ satisfying \eqref{condint}. has an asymptotic expansion 
$E_j(t)\sim\sumls{k=0}^\infty N^{-\frac{j+k}2}\cE_j^k(0)$, where $\cE_j^k(t)$ is the solution of 
%the equation 
\eqref{eqrecur} with initial condition $\cE_j^k(0)=\delta_{k,0}N^{\frac j2}E_j(0)$ (resp. $\delta_{k,1}N^{\frac j2}E_j(0)$) if $j$ is even (resp. $j$ is odd) and one has
$$
\norm{E_j(t)-\sum_{k=0}^{2n}N^{-k/2}E_j^k(t)}_1
\leq
%tC_{2n}(t)
D_{2n}(t)
N^{-n-\frac12}(
%e^{Ct\norm{V}}
%A^{2n}_t
D'_{2n}(t)
\tfrac{j^2}N)^{j/2},
$$
%for any function $\sca$ as in footnote 1 page \pageref{pagefoot1},  and  
where $D_k(t),D'_k(t)$ are defined in \eqref{dnt} below
%.
%Finally 
and satisfy,$\mbox{ as }k,|t|\to\infty$,
$$
\log{D_k(t)}=\tfrac{3k}2(\log{k}+\tfrac{|t|\norm{V}_\infty}\hbar)+O(k+\tfrac{|t|\norm{V}_\infty}\hbar)\mbox{ and }\log{D_k'(t)}=O(k+\tfrac{|t|\norm{V}_\infty}\hbar).
$$

\end{Thm}
Let us define, for $j=\unN,\ n=0,\dots$
\be\label{Enj}
%\cE_{j}^{n}(t)=\sum_{k=0}^{2n+((-1)^j-1)/2}N^{-k/2}\cE_j^k(t),\ \ 
E_{j}^{n}(t)=
%N^{j/2}\cE_{j,n}(t)=
\sum_{k=0}^{2n
%+\frac{(-1)^j-1}2
%-j
}N^{-\frac{j+k}2}\cE_j^k(t)
%=N^{[ j/2]}\sum_{\ell=[j/2]}^{-\infty}N^{\ell}\cE_j^{j-2\ell}
%=N^{[j/2]}\sum_{\ell=0}^{\infty}N^{-\ell}\cE_j^{j-2[j/2]+2\ell}
\ee
(note that $E_{j}^{n}(t)$ contains only integer powers of $N^{-1}$, since $\cE^k_j=0\mbox{ when }j+~k$ { is odd}, that is  $E^n_j=\suml_{k=[(j+1)/2]}^nc^j_kN^{-k}$), and $F^{N,n}_j(t)$ the operator of integral kernel $\FF^{N,n}_j(t)(Z_j)$ defined by
\be\label{FNnj}
%%\be\label{invdeferror}
\FF^{N,n}_j(t)(Z_j)=\sum_{k=0}^j\sum_{1\leq i_1<\dots<i_k\leq j}
\FF(t)(z_{i_1})\dots F(t)(z_{i_k})E_{j-k}^{n}(Z_{j}^{/\{i_1,\cdots,i_k\}}),\
%%\ee
\ee
(that is \eqref{invdeferror} truncated at order $n$,
 %with 
  same slight abuse of notation).
%\bea
%E_j^n&=&
%N^{j/2}\sum_{m=0}^{[n/2]}N^{-m}\cE_j^{2m}\mbox{ if $j$ is even}\nn\\
%&=&
%N^{\frac{j+1}2}\sum_{m=0}^{\frac{n+p(n)}2-1}N^{-m}\cE_j^{1+2m}
%\mbox{ if $j$ is odd},\nn
%\eea
%where $p(n)=0,\ n\mbox{ even},\ p(n)=1,\  n\mbox{ odd}$.
%(note that $j-2[j/2]=1/0$ when $j$ is even/odd).
%\bea
%E_j^n&=&
%N^{j/2}\sum_{m=0}^{n}N^{-m}\cE_j^{2m}\ \ \ \ \ \mbox{ if $j$ is even}\nn\\
%&=&
%N^{\frac{j-1}2}\sum_{m=0}^{n}N^{-m}\cE_j^{1+2m}.
%\mbox{ if $j$ is odd},\nn
%\eea
 $F^{N,n}_j$ is therefore a polynomial of order $n$ in $\frac1N$.
\begin{Cor}\label{cormain}
%Let us define
%$$
%\cE_{j,n}(t)=\sum_{k=0}^nN^{-k/2}\cE_j^k(t).
%$$
Let $\FF^N(t)$ the solution of the quantum $N$ body system \eqref{Ns}  with initial datum $\FF^N(0)=\FF^{\otimes N}$, $\FF\in\cL(L^2(\bR^{d})), \FF\geq 0,\tr{\FF}=1$, and $\FF(t)$ the solution of the Hartree equation \eqref{mfeh} with initial datum $\FF$.

Then, for all $n\geq1$ and $N\geq 4(eA^{2n}_tj)^2$,
\bea
%&&\norm{\FF^N_j(t)-\FF(t)^{\otimes j}-N^{-1}\FF(t)^{\otimes (j-1)}\cE_{j,n}
%-\sum_{\ell=2}^jN^{-\ell/2}\FF(t)^{\otimes (j-\ell)}\cE_{j,n-\ell}}\nn\\
&&
%\norm{\FF^N_j(t)-\sum_{\ell=2}^j\FF(t)^{\otimes (j-\ell)}E_{j-\ell}^{n-\ell}}
\norm{\FF^N_j(t)-
%\suml_{K\subset \{\un{j}\}}\margetracej{K}{\FF}{\{\un{j}\}}\Phi_{\{\un{j}\}/K}E_{j-|K|}^{n
%%-\frac{(j-|K|)+((-1)^{j-|K|}-1)/2}2
%}}
%\sum_{K \subset \{\un{j}\}} \Pii{K}{j}\ \big(F(t)^{\otimes |K|}\otimes E^n_{j-|K|}(t)\big)
F^{N,n}_j(t)
}_1
\leq
%2C_{2n}(t)
N^{-n-\frac12}\ 
\tfrac{2
%tC_{2n}(t)
D_{2n}(t)
e
A^{2n}_t
D'_{2n}(t)
j}{\sqrt N}.\nn
\eea
%for $N\geq 4(eA^{2n}_tj)^2$.
%%$$\suml_{K\subset J}\margetracej{K}{\FF}{J}\Phi_{J/K}E_{j-|K|}$$
%Here $C_{t,j,n}=$ is independent of $j$.
\end{Cor}
\begin{Rmk}
If one is interested only to the expansion up to order $n<j$, we can change the sum in the l.h.s. of the inequality in  Corollary \ref{cormain} by a sum up to $\ell=n$.
\end{Rmk}
\begin{proof}
The proof is similar to the one of Corollary 2.2 in \cite{PPS}.
\bea
&&\norm{\FF^N_j(t)-
%\suml_{K\subset \{\un{j}\}}\margetracej{K}{\FF}{\{\un{j}\}}\Phi_{\{\un{j}\}/K}E_{j-|K|}^n
\sum_{K \subset \{\un{j}\}} \Pii{K}{j}\ \big(F^{\otimes |K|}\otimes E^n_{j-|K|}\big)
}\nn\\
&\leq&
\sum_{k=1}^j\binom{j-k}{k}
\norm{E_{k}-E_{k}^n}\nn\\
&\leq& N^{-n-\frac12}
\sum_{k=1}^j\binom{j}{k}
C_{2n-k}(t)
(\frac{A^{2n-k}_tk^2}N)^{k/2}\nn\\
&\leq&
N^{-n-\frac12}C_{2n}(t)
\sum_{k=1}^jj(j-1)\dots(j-k+1)
(\frac{A^{2n}_t}{\sqrt N})^{k}\frac{k^k}{k!}\nn\\
&\leq&
N^{-n-\frac12}C_{2n}(t)
\sum_{k=1}^j
(\frac{jeA^{2n}_t}{\sqrt N})^{k}
%\nn\\
%&\leq&
%N^{-(n+1)}C_{2n}(t)
%\sum_{k=1}^j
%(\frac{eA^{2n}_t\sca(j)}{\sqrt N})^{k}
\leq N^{-n-\frac12}\ \tfrac{2C_{2n}(t)eA^{2n}_tj}{\sqrt N}\nn
\eea
for $N\geq 4(eA^{2n}_tj)^2$  
(we used $\frac{k^k}{k!}\leq \frac{e^k}{\sqrt{2\pi k}}$).
\end{proof}
%Let us finish this section by some bibliographic considerations.
\begin{Cor}\label{corcormain}
The rate of convergence to the meanfield limit in $\frac1N$ is optimal.
\end{Cor}
\begin{Rmk}
In the asymptotic expansion $E_j(t)\sim\suml_{k=[(j+1)/2]}^\infty c^j_k(t)N^{-k}$ the coefficients $c^j_k(t)$, such as $\cE^k_j(t)$,  depend on $N$ as well: first by the dependence of $\Delta^+_j=(1-\frac jN)C_{j+1}$ and also by $U_j(t,s)$ defined by \eqref{u0}. As mentioned already the latter can be  expressed as a (convergent) series in $\frac1N$ out of the linearization of the meanfield equation so that obtaining a full asymptotic expansion of $E_j(t)$ with the only knwoledge of the linearization of the meanfield equation is (tedious but) elementary. Let us note also that $E_j(0)$ is allowed to depend on $N$, without any restriction as soon as it satisfies \eqref{condint}.
\end{Rmk}
\section{The recursive construction}\label{recur}
Let us recall from \cite{PPS} that the hierarchy of error terms saisfy the following equation:
\bea\label{eqhieraerror}
\pa_t E_j&=& \left(K^j+\frac 1{N}T_j\right) E_j +
%\frac{N-j}N \suml_{ i\in J}  C_{i,j+1} \left(\margtrace{i}{\FF} + \margtrace{j+1}{\FF}\right) 
%D_J
%E_{J\cup\{j+1\}} \nn \\
%&+& 
%%\frac{N-j}N C_{j+1} 
%D_j^+
%E_{J\cup\{j+1\}} + 
%D_j^-
%E_{J} +
%%\frac 1N  (\suml_{i,r \in J}  T_{i,s} \margtrace{i}{\FF} \margtrace{r}{\FF}) 
%D_j^=
%E_{J}.
D_j
E_{{j}} \nn \\
&+& 
%\frac{N-j}N C_{j+1} 
D_j^1
E_{j+1} + 
D^{-1}_j
E_{j-1} +
%\frac 1N  (\suml_{i,r \in J}  T_{i,s} \margtrace{i}{\FF} \margtrace{r}{\FF}) 
D_j^{-2}
E_{j-2}.
\eea
Here the four operator $D_j,D_j^1,D^{-1}_j,D_j^{-2},\ j=0,\dots, N$,
%\bea
%D_j&:& \ \lon^j\to\lon^j\nn\\
%D^\pm_j&:& \lon^{j\pm 1}\to\lon^j\label{ds1}\\
%D^=_j&:& \ \lon^{j-2}\to\lon^J\nn
%\eea
%
are defined as follows (here again, $J=\{1,\dots,j\})$: for any operator $G\in\lone{n},n=\unN$, we denote by $G(Z_n)$ its  integral kernel, and for any function $F(Z_n),n=\unN,$ we define 
$\op{F(Z_n)}$ as being the operator on $\lone{n}$ of integral kernel $F(Z_n)$ then

%\bea
%D_j(t)E_j&=&
%\frac{N-j}N \suml_{ i\in J}  C_{i,j+1} \left(\margtrace{i}{\FF} + \margtrace{j+1}{\FF}\right) E_{J\cup\{j+1\}},\label{ddefj}\nn\\
%&=&\frac{N-j}N j\marg{j}(V_{1,j}F\otimes E_j)+\suml_{ i\in J}\marg{j+1}(V_{i,j+1}E_j\otimes F)\nn\\
%D^+_j E_{j+1}&:=&\frac{N-j}N C_{j+1} E_{J\cup\{j+1\}},\label{defd+}\nn\\
%D_j^{-} (t) E_{j-1}&:=&{(}-\frac jN  \suml_{i\in J} \margtrace{i}{Q(\FF,\FF)} + \frac 1{2N} \suml_{i,r \in J}  T_{i,r} \margtrace{i}{\FF}\big) E_{J},  \label{defd-}\nn\\
%D_j^{=}(t)  E_{j-2}&:=&\frac 1N  (\suml_{i,r \in J}  T_{i,s} \margtrace{i}{\FF} \margtrace{r}{\FF}) E_{J},\nn\label{defd=}
%\eea
\bea
D_j&:& \ \lone{j}\to\lone{j}\nn\\
&&E_j\mapsto\frac{N-j}N \suml_{ i\in J}  C_{i,j+1}
 \left(
%\margetracej{i}{\FF}{J\cup\{j+1\}}  \Phi_{{(J\cup\{j+1\})/\{i\}}}E_j
%+ \margetracej{j+1}{\FF}{J\cup\{j+1\}}E_j
\op{F(z_i)E_j(Z_{j+1}^{/\{i\}})}
+
\op{F(z_{j+1})E_j(Z_j)}
\right)\nn\\
&&\hskip 1.2cm
-\frac1N
 \suml_{i\neq l\in J }
 %\margetracej{\{l\}}{\FF}{J} C_{i,j+1}\Phi_{(J/\{l\})\cup \{j+1\}}E_{j}
%\op{F(z_l)\big(C_{i,j+1}(\op{F(z_i)E_j(Z_{j+1}^{/\{l\}})}\big)(Z_j^{/\{l\}})}
C_{i,j+1}(\op{F(z_l)E_j(Z_{j+1}^{/\{l\}})}
\label{ddefj}\\
D_j^1&:& \lone{(j+ 1)}\to\lone{j}\nn\\
&& E_{j+1}\mapsto
\frac{N-j}N C_{j+1} 
%\Phi_{J\cup\{j+1\}}
E_{j+1}\nn\\ \label{defd+} \\
D^{-1}_j&:& \lone{(j- )1}\to\lone{j}\nn\\
&& E_{j-1}\mapsto
%\big(\frac 1{N} \suml_{i,r \in J}  T_{i,r} 
%\margetracej{i}{\FF}{J}-\frac jN  \suml_{i\in J} \margetracej{i}{Q(\FF,\FF)}{J} \big) \Phi_{J/\{i\}}E_{j-1} 
\frac 1{N} \suml_{i,r \in J}  T_{i,r}
\op{F(z_i)E_{j-1}(Z_j^{/\{i\}})}
-
\frac jN  \suml_{i\in J}
\op{Q(\FF,\FF)(z_i)E_{j-1}(Z_j^{/\{i\}})}
\nn\\ 
&&\hskip 1.5cm-
\frac1N
\suml_{i\neq l\in J }
%\margetracej{\{l\}}{\FF}{J} 
%C_{i,j+1} \margetracej{\{j+1\}}{\FF}{(J/\{l\})\cup\{j+1\}} 
%\Phi_{J/\{l\}}E_{j-1}
%\op{
%C_{i,j+1} (\op{F(z_l)F(z_{j+1})
%E_{j-1}(Z_j^{/\{l\}}})(Z_{j})}
{
C_{i,j+1} (\op{F(z_l)F(z_{j+1})
E_{j-1}(Z_j^{/\{l\}}})}
 \nn\\
&&\hskip 1.5cm-
\frac1N
\suml_{i\neq l\in J }
%\margetracej{\{l\}}{\FF}{J}
%C_{i,j+1}\margetracej{\{i\}}{\FF}{(J/\{l\})\cup\{j+1\}}\Phi_{(J/\{i,l\})\cup 
%\{j+1\}}E_{j-1}
%\op{
%C_{i,j+1}\big(\op{F(z_l)F(z_i)E_{j-1}(Z_{j+1}^{/\{i,l\}}}\big)
%(Z_j)}
{
C_{i,j+1}\big(\op{F(z_l)F(z_i)E_{j-1}(Z_{j+1}^{/\{i,l\}}}\big)
}
\label{defd-}\\
D_j^{-2}&:& \ \lone{(j-2)}\to\lone{j}\nn\\
&&E_{j-2}\mapsto
\frac 1N  \suml_{i,s \in J}  T_{i,s} 
%\margetracej{i}{\FF}{J} \margetracej{r}{\FF}{J/\{i\}} \Phi_{J/\{i,s\}}E_{j-2}
\op{F(z_i)F(z_r)E_{j-2}(Z_j^{/\{i,r\}})}
.\nn\\
&&\hskip 1.6cm 
-\frac1N
\suml_{i\neq l\in J } 
%\margetracej{\{i\}}{Q(\FF,\FF)}{J}
%\margetracej{\{l\}}{F}{J/\{i\}}
%\Phi_{J/\{i,l\}}E_{j-2}
\op{
Q(\FF,\FF)(z_i)F(z_l)
E_{j-2}(Z_l^{/\{i,l\}})}.\label{defd=}
\eea
where, by convention, 
%\bea
%\label{E0}
%&& D_N^1=D_1^{-2}=0\nn\\
%&& D_1^{-1} \left(E_0\right):=-\frac 1N Q(F,F)\;,
%% \quad D_1^{-2}\left(E_{-1}\right):=0\;,
%\\
%&& D_2^{-2}\left(E_0\right):= \frac 1 N  T_{1,2} F \otimes F - \frac{1}{N} Q(F,F)^{\otimes\{1\}} 
%F^{\otimes \{2\}} -  \frac{1}{N} Q(F,F)^{\otimes\{2\}} 
%F^{\otimes \{1\}} \;.\nn
%\eea
\be
\label{E0q}
\left\{
\begin{array}{l}
 D_N^1:=D_1^{-2}:=0\\
 D_1^{-1} \left(E_0\right):=-\frac 1N Q(F,F)\;,
% \quad D_1^{-2}\left(E_{-1}\right):=0\;,
\\
D_2^{-2}\left(E_0\right):= \frac 1 N  \big(T_{1,2} (F \otimes F) -  Q(F,F)\otimes 
F -  F\otimes Q(F,F) 
\big) \;.
\end{array}
\right.
\ee
In \eqref{ddefj}-\eqref{defd=}, $F(z)$ is meant as being the integral kernel of $F(t)$ solution of the Hartree equation \ref{mfeh}.

%Note that one has the following estimates:
%\be\label{endowq}
%\| D_j \| , \| D^1_j\| \leq j  \mbox{ and }  \| D_j^{-1}\|,
%\| D_j^{-2}\|, \|D^{-1}_1(E_0)\|,  \|D^{-2}_2(E_0)\|
% \leq \frac {j^2} N.
%\ee
%\begin{Rmk}\label{crucial}We shall see that the  properties \eqref{E0q}-\eqref{endowq}, together with \eqref{normpropag}, are actually the only ones being used in the proof of Theorem \ref{main}.
%%: 
%%
%%\begin{itemize}
%%%\item
%%%$\| E_{j} \| \leq 
%%%%\suml_{K\subset J_0=\{1,\dots,j_0\}}
%%%\mbox{card}(\{K\subset J_0=\{1,\dots,j\}\})
%%%=2^{j}.
%%%$
%%\item 
%%$D_N=D^+_N=D^-_0=D^=_1=D^=_0=0$
%%\item 
%%$\norm{D_j^{-,=}}\leq\frac32\frac{j^2}N\norm{V}$
%%\item
%%$\norm{D_j^{+,-,=}}\leq j\norm{V}$
%%\end{itemize}
%\end{Rmk}
We define $H_j(t)=K^j+T_j/N+D_j(t)$ and recall the definition of the (two parameters) semigroup $U_j (t,s)$ 
%for $s \leq t$, 
satisfying, for all $s,t\in\bR$,
\bea\label{u0}
\ \ \ \ \ \ \pa_t U_j (t,s)&=&
%(K^J+\frac {T_J}{N}+D_j(t)) 
H_j(t)
U_j (t,s),\ j=1,\dots,N\nn \\
 \ \ \ U_j (s,s)&=&I=:U_0 (t,s) 
\eea 
%Obviously $U_j (t,s)U_j (s,t)=I$ so that
%\be
%U_j (t,s)^{-1}=U_j (s,t).
%\ee
%Having in mind the result of Theorem 2.1. of \cite{PPS} and calling $E_j=N^{-j/2}\cE_j,
%%\ j\geq2,\ \cE_1=N^{-1}E_1
%$,  $H_j=K^j+T_j/N+D_j$, 
%%$\cE_j=N^{\frac j2}{E_j}{}$,
% $\Delta_j^+=D_j^1$, $\Delta_j^-=ND^{-1}_j, \Delta^=_j:=ND_j^{-2}$ 
 
 Let us perform the following rescaling
 \be\label{resc}
 \left\{
 \begin{array}{rcl}
 E_j&=&N^{-j/2}\cE_j\\
\Delta_j^+&=&D_j^1\\
\Delta_j^-&=&ND^{-1}_j\\
\Delta^=_j&=&ND_j^{-2}
 \end{array}
 \right.
 \ee
 
 We find easily that, again with the convention \eqref{E0q},

\bea
%\partial_t\cE_j&=&H_j\cE_j+N^{-\frac12}D^+_j\cE_{j+1}
%+N^{\frac12}D^-_j\cE_{j-1}+ND^=_j\cE_{j-2}, j>2\nn\\
%\partial_t\cE_1&=&
%H_1\cE_1+N^{-\frac12}\Delta_1^+\cE_2
%+N^{-\frac12}\Delta^-_1\cE_0\nn\\
%&=&H_1\cE_1+N^{-\frac12}\Delta_1^+\cE_2
%-N^{-\frac12}Q(\FF,\FF)
%%+\cE_1^-
%%+O(N^{-\frac12})
%\nn\\
\partial_t\cE_j&:=&H_j\cE_j+N^{-\frac12}D_j^1\cE_{j+1}
+N^{\frac12}D^{-1}_j\cE_{j-1}+ND_j^{-2}\cE_{j-2}\nn\\
&=&
H_j\cE_j+N^{-\frac12}\Delta^+_j\cE_{j+1}
+N^{-\frac12}\Delta^-_j\cE_{j-1}+\Delta^=_j\cE_{j-2}
%, j\geq 2
\label{theeq}\\
%\mbox{Therefore }&&\nn\\
%\partial_t\cE_1&=&H_1\cE_1+O(N^{-\frac12})\nn\\
%\partial_t\cE_2&=&H_2\cE_2+\Delta^=_2\cE_0+O(N^{-\frac12})\nn\\
%&=&
%H_2\cE_2+V^2\FF\otimes\FF+O(N^{-\frac12})\nn\\
%&.&\nn\\
%%&.&\nn\\
%%&.&\nn\\
%\partial_t\cE_j&=&
%%H_j\cE_j+D_j^=\cE_{j-2}+D_j^+\cE_{j+1}
%%+D_j^-\cE_{j-1}\nn\\
%%&=&H_j\cE_j+\cE_j^=+\cE_j^++\cE_j^-\nn\nn\\
%%&=&
&=&H_j\cE_j+\Delta^=_j\cE_{j-2}+O(
%\sca(j)
N^{-\frac12})
%, j\geq 3
.\nn
\eea
%Replacing $\cE_j^=$ by their values, w
We define $\cE^0_j(t)$ as the solution of 
\bea
\partial_t\cE_1^0&=&H_1\cE_1^0
%+D_1^+\cE_2^0-Q(\FF,\FF)
\label{e11}\\
\partial_t\cE_2^0&=&
%H_2\cE_2^0+D_2^=\cE_0=
H_2\cE_2^0+
%V^2\FF^{\otimes 2}
T_{1,2} (F \otimes F) -  Q(F,F)\otimes 
F -  F\otimes Q(F,F)
\label{e22}\\
%&.&\nn\\
%&.&\nn\\
%&.&\nn\\
\partial_t\cE_j^0&=&
%H_j\cE_j^0+\cE_j^=
%=
H_j\cE_j^0+\Delta_j^=\cE_{j-2}^0, j\geq 3.
\label{ej}
\eea 
\eqref{e11} and \eqref{e22}  are two closed equations 
whose solutions are given by 
\be\label{solu01}
\cE_1^0(t)=U_1(t,0)\cE_1^0(0)=0
\ee
since we supposed $E_1(0)=O(N^{-1})$, and
\bea\label{solu02}
\cE_2^0(t)&=&U_1(t,0)\cE_2^0(0)\nn\\
&+&U_2(t,0)\int_0^tU_2(0,s)(T_{1,2} (F \otimes F) -  Q(F,F)\otimes 
F -  F\otimes Q(F,F))ds.
\eea
%The solution of \eqref{e2} gives the solution of \eqref{e1} by
%\be
%\cE_1^0(t)=\cE_1^0(0)+U_1(t,0)\int_0^t U_1(0,s)(D_1^+\cE_2^0(s)-Q(\FF(s),\FF(s)))ds
%\\
%\ee
Iterating till $j$, we get explicitly the solution of \eqref{e11}-\eqref{ej} given by
\be\label{ind1}
\cE_j^0(t)=U_j(t,0)\cE_j^0(0)+U_j(t,0)\int_0^t U_j(0,s)\Delta_j^=\cE_{j-2}^0(s)ds, j\geq 1,
\ee
with the convention $\cE^k_l=0,l<0$.

Therefore, for $j=\unN, t\in\bR$, the knowledge of $U_{j}(t,s),\  |s|\leq |t|,$,  and $\cE_{j'}^0(0),\ j'=1,\dots,j$ guarantees  the knowledge of $\cE^0_{j'}(t),t\in\bR,j'\leq j$. We write this fact as
\be\label{leads1}
%\{\cE^0_{j_0}(0),j_0=\un{j}\}\leadsto \{\cE^0_{j_0}(t),\ t\geq 0,\ j_0\}=\un{j}\}
(\cE^0_{j'}(0))_{j'=\un{j}}\leadsto (\cE^0_{j'}(t))_{ t\in\bR, j'=\un{j}}
\ee
\vskip 1cm
Making now the ansatz $\cE_j(t)\sim\suml_{k=0}^\infty\cE_j^k(t)N^{-k/2}$ we find that the family $(\cE^k_j(t))_{j=\unN,k=0,\dots}$ must satisfy
%To go one step further, one find
%\bea
%\partial_t\cE_1^1&=&H_1\cE_1^1+\Delta_1^+\cE_2^0 + \Delta_1^-\cE_0^0\label{e11}\\
%&&+\Delta_1^+\cE_2^0-Q(\FF,\FF) \nn\\
%\partial_t\cE_2^1&=&H_2\cE_2^1+\Delta_2^+\cE_3^0+\Delta_2^-\cE_1^0 \label{e21}
%\\
%&\cdot &\nn\\
%\partial_t\cE_j^1&=&H_1\cE_j^1+\Delta_j^=\cE_{j-2}^1+\Delta_j^+\cE_{j+1}^0+\Delta_j^-\cE_{j-1}^0,j>2   \label{e11}
%\eea
%Recursively we find:
\bea\label{eqk}
%\partial_t\cE_1^k&=&H_1\cE_1^k+\Delta_1^+\cE_2^{k-1}+\Delta^-_1\cE_0^{k-1}   \label{e11k}\\
%\partial_t\cE_2^k&=&H_2\cE_2^k+\Delta_2^+\cE_3^{k-1}+\Delta_2^-\cE_1^{k-1} \label{e21k}\\
%&\cdot &\nn\\
\partial_t\cE_j^k&=&H_j\cE_j^k+D_j^{-2}\cE_{j-2}^k+\Delta_j^+\cE_{j+1}^{k-1}+\Delta_j^-\cE_{j-1}^{k-1},   \label{e11k}
\eea
%where $\Delta^+=D^+,\ \Delta^-=ND^-$, 

with again the conventions \eqref{E0q} and  $\cE^l_j=0$ for $l<0$, solved %respectively 
by
\bea
%\cE_1^k(t)&=&U_1(t,0)\cE_1^k(0)+
%U_1(t,0)\int_0^tU_1(0,s)(\Delta_1^+\cE_2^{k-1}+\Delta^-_1\cE_0^{k-1}  )ds\nn\\
%&& \label{solu1}\\
%\cE_2^k(t)&=&U_2(t,0)\cE_2^k(0)+
%U_2(t,0)\int_0^tU_2(0,s)(\Delta_2^+\cE_3^{k-1}(s)+\Delta_2^-\cE_1^{k-1}(s))ds\nn\\
%&& \label{solu2}\\
\cE_j^k(t)&=&U_j(t,0)\cE_j^k(0)\nn\\
&+&\int_{s=0}^tU_j(t,s)(\Delta_j^=\cE_{j-2}^k(s)+\Delta_j^+\cE_{j+1}^{k-1}(s)+\Delta_j^-\cE_{j-1}^{k-1}(s))ds   \label{soluk}
%&=&\cE_j^k(0)\nn\\
%&+&\int_{s=0}^t\widetilde U_j(t,s)(D_j^+\cE_{j+1}^{k-1}(s)+D_j^-\cE_{j-1}^{k-1}(s))ds   \label{soluktilde}
\eea
%where $\widetilde U_j(t,s):\lono{\leq j}\to\lono{ j},\ \lono{\leq j}:=
%\oplus_{n=1}^j\lono{n}$, is defined 
%%in Section \ref{esti}
%by:
%\bea
%\partial_t\widetilde U_j(t,s)&=&(K^j+\tfrac{T_j}N+D_j^=(t))\widetilde U_j(t,s)\nn\\
%\widetilde U_j(s,s)&=&I\nn\\
%\widetilde U_j(t,s)&=&\sum_{n=0}^{[j/2]}\int_0^tU^0_j(t,t_1)D^=_j(t_1)\nn\\
%& &\int_0^{t_1}\dots\int_0^{t_n}U^0_{j-2}(t_1,t_2)D^=_{j-2}(t_2)\dots D^+_{j-2(n-1)}(t_n)\label{defutilde}\\
%&& dt_1\dots dt_n.\nn
%\eea
Since 
$\cE^k_{-1}(t)=0$ by convention and $\cE^k_0(t)=0$ for $k\geq 1$ since $E_0(t):=1$, we find after \eqref{leads1} that $\cE_1^1(t)$ and  $\cE_2^1(t)$ are determined by $\cE_1^1(0)$ and  $\cE_2^1(0)$. Therefore, by \eqref{soluk}, $\cE_j^1(t), j=\unN$ are determined by $(\cE_j^1(0))_{j=\unN}$, and determine  $\cE_1^2(t)$ and $\cE_2^2(t)$. These ones determine in turn all the $\cE_j^2(t),j=\unN$ and so on.

%and to \eqref{solu01},\eqref{solu02} and \eqref{ind1} for $k=0$. Moreover, although $\cE^0_0\in\lone{0}=\bC\mbox{ and so }\notin\lone{1}$, $\Delta^-_1\cE^0_0\in\lone{1}$ and (see \cite{PPS}, (54))
%\be\label{delat10}
%\Delta^-_1\cE^0_0(s)=-Q(F(s),F(s)) \mbox { and }\norm{\Delta^-_1\cE^0_0(s)}\leq 1.
%\ee 
\vskip 0.5cm
Therefore, the knowledge of $(\cE^{k'}_{j'}(s))_{|s|\leq |t|,k'\leq k-1,j'=\un{j+1}}$ and $\cE_j^k(0)$ guarantees fpr all $j,k$, by induction, the knowledge of $\cE_j^k(t)$. Thus
%\eqref{solu1}.
%,\eqref{solu2}, the knowledge of $\cE^k_1(t),\cE^k_2(t)$ and therefore, by \eqref{soluk},  the knowledge 
%%of $(\cE^{k-1}_j(t))_{j=\unN}$ guarantees, by \eqref{solu2}, the knowledge 
%of $\cE^k_3(t)$, $\cE^k_4(t)$ and so forth till $\cE^k_j(t)$:
\be\nn
\big((\cE_{j}^{k}(0),(\cE^{k'}_{j'}(s))_{|s|\leq |t|,k'\leq k-1,j'=\un{j+1}}\big)\leadsto (\cE^{k'}_{j'}(s))_{|s|\leq |t|,k'\leq k,j'=\un{j}}.
\ee
Therefore, supposing known $(\cE_{j'}^{k'})_{k'\leq k,j'\leq j}$,
\be\nn
(\cE^{k'}_{j'}(s))_{s\leq t,k'\leq k-2,j'=\un{j+2}}\leadsto 
(\cE^{k'}_{j'}(s))_{s\leq t,k'\leq k-1,j'=\un{j+1}}\leadsto \cE^k_j(t).
\ee
and by iteration
\be\nn
(\cE^{0}_{j'}(s))_{s\leq t,j'=\un{j+k}}\leadsto \cE^k_j(t)
\ee
so that, by \eqref{leads1},
\be\nn
(\cE^{0}_{j'}(0))_{j'=\un{j+k}}\leadsto \cE^k_j(t).
\ee
%. Since we know the $(\cE^{k-1}_j(t))_{j=\unN}$ 
%% true
%  for $k=1$ by \eqref{ind1}, we have determined all the $(\cE^{k}_j(t))_{j=\unN, k=1\dots}$
%\be\nn
%\{\cE_1^{k-1},\dots,\cE_{j+1}^{k-1},\cE_{j-2}^k\}\Longrightarrow \cE_j^k\ \mbox{ by \eqref{soluk}}
%\ee
%and
%\be\nn
%\{\cE_1^{k-1},\dots,\cE_{j+1}^{k-1}\}\Longrightarrow \cE_j^k\ \mbox{ by \eqref{soluktilde}}.
%\ee
\vskip 1cm
We just proved the following result.
\begin{Prop}\label{asympt}
For any $j=\unN,t\geq 0,k=0,\dots$, let $\cE^k_j(t)$ be the solution of \eqref{soluk}. Then $\cE^k_j(t)$ is determined by the values $\cE^{k'}_{j'}(0)$ for $0\leq k'\leq k$, $1\leq j'\leq j+k$.
%The formulas \eqref{soluk} for $j=\unN$ and $k=1,\dots$ determine an asymptotic expansion of the solution of the system of equations \eqref{theeq} for $j=\unN$. 
\end{Prop}

Formula \eqref{soluk} will give easily the following result.
\begin{Prop}
Let $\cE^{k'}_{j'}(0)=0$ for $j'\leq j, k'\leq k,j'+k'$  odd. Then $\cE^k_j(t)=0$ for $j+k$ odd.
%$j$ and $k$ have inverse parity\footnote{$j$ even, $k$ odd or $j$ odd, $k$ even}.
\end{Prop}
\begin{proof}
Let us suppose $\cE^{k'}_{j'}(0)=0$ for $j'\leq j, k'\leq k,j'+k'$  odd.
%Let us write \eqref{soluk} again
%\be
%\cE_j^k(t)=\int_{0}^tU_j(t,s)(\Delta_j^=\cE_{j-2}^k(s)+\Delta_j^+\cE_{j+1}^{k-1}(s)+\Delta_j^-\cE_{j-1}^{k-1}(s))ds   \label{solukbis}
%\ee
By \eqref{e11} we have that $\cE_1^0(t)=0$ since $\cE_1^0(0)=0$. 
%Since \eqref{solukbis} gives 
%$\cE_j^k(t)=\int_{s=0}^tU_j(t,s)\Delta_j^=\cE_{j-2}^0(s)ds$, we get 
Therefore, by induction on $j$ in \eqref{soluk}, $\cE^0_j(t)=0$ for all $j$ odd.

\noindent Since $\cE_0(t):=1,\ \cE^j_0(t)=0, j>0$,  so that $\cE^1_2(t)=0$ by \eqref{soluk} and therefore 
$\cE^1_j(t)=0$ for all $j$ even, since then $j\pm 1$ is odd, and therefore $\cE^0_{j\pm 1}(s)=0$ .
This gives $\cE^2_1(t)=0$ by \eqref{soluk} and so on.
\end{proof}
\begin{Cor}
$E_j(t)$ has an asymptotic expansion in powers of $N^{-1}$ with leading order $N^{-[\tfrac{j+1}2]}$ when $E_j(0)$ is so.
\end{Cor}
\section{Estimates and proof of Theorem \ref{main}}\label{esti}

%We will take for simplicity $\cE_j^k(0)=\delta_{j,0}\delta_{k,0}, \forall k,\forall j$.

%Theorem 2.1 of \cite{PPS} expressed on the quantities $\cE_j$ reads: let $
% \cE_j$ satisfy
In order to simplify the expressions, we will first suppose that $\tfrac{\norm{V}_{L^\infty}}\hbar= 1$.
Note that one has therefore the following estimates:
\be\label{endowq}
\| D_j \| , \| D^1_j\| \leq j  \mbox{ and }  \| D_j^{-1}\|,
\| D_j^{-2}\|, \|D^{-1}_1(E_0)\|,  \|D^{-2}_2(E_0)\|
 \leq \frac {j^2} N.
\ee

Let us first notice that \eqref{eqhieraerror111} expressed on the $\cE_j$s reads
\be\label{hypthm21}
\partial_t\cE_j=
H_j\cE_j+N^{-\frac12}\Delta^+_j\cE_{j+1}
+N^{-\frac12}\Delta^-_j\cE_{j-1}+\Delta^=_j\cE_{j-2}
\ee
and that \eqref{condint} and \eqref{timet} can be rephrased as
\be\label{thm21}
\norm{\cE_j(0)}\leq (Aj^2)^{j/2}\Longrightarrow
\norm{\cE_j(t)}\leq (A_tj^2)^{j/2},\ A_t=C'Ae^{Ct}
\ee
for some explicit constants $A',C$.

%Let us consider the solution at order $k$ of \eqref{hypthm21}
% %with nul initial data 
% given by \eqref{soluk}:
%\be
%\cE_j^k(t)=\int_{s=0}^tU_j(t,s)(\Delta_j^=\cE_{j-2}^k(s)+\Delta_j^+\cE_{j+1}^{k-1}(s)+\Delta_j^-\cE_{j-1}^{k-1}(s))ds +U_j(t,0)\cE^k_j.  \nn
%\ee
We get
\bea
\partial_t\cE_j^k(t)
&=&
H_j(t)\cE_j^k(t)+\Delta_j^=\cE_{j-2}^k(t)+N^{-\frac12}(\Delta_j^+\cE_{j+1}^{k-1}(t)+\Delta_j^-\cE_{j-1}^{k-1}(t))\nn\\
&=&
H_j(t)\cE_j^k(t)+\Delta_j^=\cE_{j-2}^k(t)+\Delta_j^+\cE_{j+1}^{k}(t)+\Delta_j^-\cE_{j-1}^{k}(t))\nn\\
&&
+N^{-\frac12}((\Delta_j^+(\cE_{j+1}^{k-1}(t)-\cE_{j+1}^{k}(t))+(\Delta_j^-(\cE_{j-1}^{k-1}(t)
-\cE_{j-1}^{k}(t))),\nn
\eea
\newcommand{\bcE}
{\bar{{\cE}}}
%{{\cE^\leq}}
and, calling $\bcE^{ n}_j=\suml_{k=0}^nN^{-k/2}\cE^k_j$, one easily check that
\bea
\partial_t\bcE^n_j(t)
&=&
H_j(t)\bcE_j^n(t)+\Delta_j^=\bcE_{j-2}^n(t)+N^{-\frac12}(\Delta_j^+\bcE_{j+1}^{n}(t)+\Delta_j^-\bcE_{j-1}^{n}(t))\nn\\
&&
-
 N^{-\frac{n+1}2}(\Delta_j^+(\cE_{j+1}^{n}(t))+\Delta_j^-(\cE_{j-1}^{n}(t)).\label{hgfdd}
\eea
\newcommand{\bbcE}{R}
Therefore $\bbcE^n_j:=\cE_j-\bcE^n_j$ 
%and taking $\bbcE^n_j(0)=0$,
satisfies the equation
\bea
\partial_t\bbcE^n_j(t)
&=&
H_j(t)\bbcE_j^n(t)+\Delta_j^=\bbcE_{j-2}^n(t)+N^{-\frac12}(\Delta_j^+\bbcE_{j+1}^{n}(t)+\Delta_j^-\bbcE_{j-1}^{n}(t))\nn\\
&&
+
 N^{-\frac{n+1}2}(\Delta_j^+(\cE_{j+1}^{n}(t))+\Delta_j^-(\cE_{j-1}^{n}(t))\label{bhgfdd}
\eea
%
%On $\overset{N}{\underset{j=1}\oplus}\lone{j}$ we 
Let us define the mapping
%the two parameters semigroup $$U^N(t,s): (\cE_j(s))_{j=\unN}\mapsto U^N(t,s)[(\cE_j(s))_{j=\unN}]:=(\cE_j(t))_{j=\unN},$$
% and let us call
 $$U^N_j(t,s): (\cE_j(s))_{j=\unN}\mapsto U^N_j(t,s)\big((\cE_j(s))_{j=\unN}\big):=\cE_j(t).$$
 %such that $U^N_j(t,s)=P_{\lone{j}}U^N(t,s)$ solves \eqref{hypthm21}, 
 %Namely,
 In other words, the family 
 $(U^N_j(t,s))_{j=\unN} $ solves the equation:
\bea
\partial_tU^N_j(t,s)
&=&
H_j(t)U^N_j(t,s)+\Delta_j^=U^N_{j-2}(t,s)\nn\\
&&+N^{-\frac12}(\Delta_j^+U^N_{j+1}(t,s)+\Delta_j^-U^N_{j-1}(t,s)),\nn\\
 U^N_j(s,s)&=&I
.\nn
\eea
Hence, the solution of \eqref{bhgfdd} reads
\bea\label{houlala}
\bbcE^n_j(t)
&=&U^N_j(t,0)((\bbcE^n_j(0))_{j=\unN})\nn\\
&+&N^{-\frac{n+1}2}\int_0^tU^N_j(t,s)
%[
((\Delta_j^+(s)\cE_{j+1}^{n}(s))+\Delta_j^-(s)\cE_{j-1}^{n}(s))_{j=\unN})
%_{j=\unN}]
ds
%+U^N_j(t,0)(\bbcE^n_j(0))
\eea
with again the same convention on negative indices.
%here
%%, in \eqref{houlala}, 
%$\Delta_j^+(s)\cE_{j+1}^{n}(s))+\Delta_j^-(s)\cE_{j-1}^{n}(s)$ (resp. $\bbcE^n_j(0)$) is meant for the sequence $\big((\Delta_j^+(s)\cE_{j+1}^{n}(s))+\Delta_j^-(s)\cE_{j-1}^{n}(s))\delta_{k,j}\big)_{k=\unN}$ (resp. $(\bbcE^n_k(0)\delta_{k,j})_{k=\unN}$).

By hypothese, $\bbcE_j^n(0)=0$ since $\cE_j^n(0)=\delta_{n,0}\cE^0_j(0)$.

Let us suppose now that 
%$\cE_j(0)$ has an asymptotic expansion of the form $\sum\limits_0^\infty N^{-k/2}\cE^k_j(0)$ 
%%with $\norm{\cE^k_j(0)}=O((Aj^2)^{j/2})$ so 
%such 
%that $\norm{\bbcE^n_j(0)}\leq N^{-\frac{n+1}2}A_n(Aj^2)^{j/2}$ for some positive constants $A_n,A$. Then \eqref{thm21} implies that
%\be\label{ouf1}
%\norm{U^N_j(t,0)((\bbcE^n_j(0))}_{j=\unN})\leq N^{-\frac{n+1}2}D_n(C'Ae^{C|t|}j^2)^{j/2}.
%\ee
%%Therefore, Theorem 2.1.in \cite{PPS}, i.e. \eqref{hypthm21}-\eqref{thm21}, expresses 
%By the same idea we get
%that, if 
\be\label{e1}
\norm{
\Delta_j^+(\cE_{j+1}^{n}(s))+\Delta_j^-(\cE_{j-1}^{n}(s)
}\leq C_n(s)(C'_n(s)j^2)^{j/2},\ \mbox{ }|s|\leq |t|,
\ee
 for two increasing functions 
 %of $s$, $ C^n_s\geq 1$ and 
 $C_n(s),C'_n(s),C'_n(s)\geq1$,
 %(say), 
  Then \eqref{thm21} implies that
  $$
  \norm{U^N_j(t,s)
((\Delta_j^+(s)\cE_{j+1}^{n}(s))+\Delta_j^-(s)\cE_{j-1}^{n}(s))_{j=\unN})
}
%\nn\\&&
\leq 
%(\frac1{\sqrt N})^{n+1}
%N^{-\frac{n+1}2}
C_n(s)
(C'C'_n(s)e^{C|t|}j^2)^{j/2},\nn
$$
and thus
  \bea
  \norm{\cE_j(t)-\bcE^n_j(t)}&=&\norm{\bbcE^n_j(t)}\nn\\
  &=&\norm{\int_0^tU^N_j(t,s)
((\Delta_j^+(s)\cE_{j+1}^{n}(s))+\Delta_j^-(s)\cE_{j-1}^{n}(s))_{j=\unN})
ds}\nn\\
%&&
%\leq 
%%(\frac1{\sqrt N})^{n+1}
%%N^{-\frac{n+1}2}
%tC_n(t)
%(C'C'_n(t)Ae^{C|t|}j^2)^{j/2},\nn
%\eea
%and
%\bea\label{e2}
%\norm{\cE_j(t)-\bcE^n_j(t)}=\norm{\bbcE^n_j(t)}
&\leq&
%(\frac1{\sqrt N})^{n+1}
N^{-\frac{n+1}2}
D_n(t)
(
%C'C'_n(t)Ae^{C|t|}
D'_n(t)\nn
j^2)^{j/2},
\eea
where
%, since $A,C'_n(s)\geq1 $,
\be\label{dnt}
D_n(t)=
%2\sup{(D_n,
tC_n(t)
%)}
\mbox{ and } D'_n(t)=C'C'_n(t)e^{C|t|}.
\ee
%($(A^n_t)'$ depends only on $A^n_t$, and on the same way that $A'$ depends on $A$ in \eqref{thm21}.

It remains to prove an estimate like \eqref{e1}. 

We will obtain such an estimate by iterating \eqref{soluk}. We first remark that, since $e^{K^j+T_j/N}$ is unitary and $\norm{D_j}\leq j$, the Gronwall Lemma gives that
\be\label{gro}
\norm{U_j(t,s)}\leq e^{j|t-s|}.
\ee
%Using now
We will use
\bea\label{readily}
&&\prod_{i=0}^{m} e^{ (j+i) (t_{i}-t_{i+1}) }  \leq e^{(j+m)|t_{m+1}-t_0|}\ \mbox{ for any }(t_i)_{i=0,\dots,m} \mbox{ (see \cite{PPS}}),\label{hyp1}\\
&&\norm{\Delta^\pm},\norm{\Delta^=}\leq j^2,\label{hyp2}\\
%&& \norm{\cE_{j'}(t)}\leq (C'Ae^{C|t|}(j')^2)^{(j')/2}\leq (C'Ae^{C|t|}(j+n)^2)^{(j+n)/2}, j'\leq j+n,\nn\\
%&&\mbox{ and }\nn\\
&& \int_0^tdt_1\int_0^{t-1}dt_2\dots\int_0^{t_{n-1}}dt_n=\frac{t^n}{n!}.\label{hyp4}
\eea
Let us remind that we have $\cE^k_0(t)=\delta_{k,0}$ for all $t$ $\cE_j^k(0)=\delta_{k,0}N^{\frac j2}E_j(0)$ (resp. $\delta_{k,1}N^{\frac j2}E_j(0)$) if $j$ is even (resp. $j$ is odd) with $\norm{E_j(0)}_1\leq (A\tfrac{j^2}N)^{j/2}$. 
%Moreover, let us consider the solution of \eqref{hypthm21} with initial condition $\widetilde\cE_j^k(0)=\tilde\delta_{j,k}\cE_j(0)$ where $\tilde\delta_{j,k}=\delta_{k,0}$ (resp. $\delta_{k,0}$) if $j$ is even (resp. odd).
%Without loss of generality, we might take $\cE_j^k(0)=\delta_{k,0}\cE_j(0)$ (possibly depending smoothly on $1/N$) such that $\norm{\cE^0_j(0)} \leq (Aj^2)^{j/2},\ A\geq 1$.
%Under this hypothesis (in fact a way of regrouping the dependence in $N$ of the initial data), 
Therefore \eqref{soluk} reads:
$$
\left\{
\begin{array}{rcl}
\cE^0_j(t)&=&U_j(t,0)\cE^0_j(0)+\int_{s=0}^tU_j(t,s)\Delta^=_j\cE^0_{j-2}(s)ds,\ \cE^1_j(t)=0,\ \mbox{ j even}\nn\\
\cE_j^1(t)&=&U_j(t,0)\cE^1_j(0)+\int_{s=0}^tU_j(t,s)\Delta_j^=\cE_{j-2}^1(s)ds\nn\\
&&+\int_0^tU_j(t,s)(\Delta_j^+\cE_{j+1}^{0}(s)+\Delta_j^-\cE_{j-1}^{0}(s))ds,\ \cE^0_j(t)=0,\ \mbox{ j odd},   \nn
\\
\cE_j^k(t)&=&\int_{0}^tU_j(t,s)(\Delta_j^=\cE_{j-2}^k(s)+\Delta_j^+\cE_{j+1}^{k-1}(s)+\Delta_j^-\cE_{j-1}^{k-1}(s))ds,\ k>1.   \nn
\end{array}
\right.
$$
Let us note first that \eqref{e11k} for $k=0$ is verbatim \eqref{eqhieraerror111} after replacing $E_j$ by $\cE_j^0$ and $D_j^{\pm}$ by $0$. On the other side, we know by Remark 3.2 in \cite{PPS}, that the proof of Theorem 2.1 in \cite{PPS}, Theorem \ref{main} in the present paper,  depends on $D_j^{\pm}$ only through its norm $\norm{D_j^{\pm}}$ required to be bounded by $j^2$. Therefore we get immediately, for $j$ even,
%We get
%\bea\nn
%\cE^0_j(t)-U_j(t,0)\cE^0_j(0)&=&
%\sum_{k=1}^{[j/2]-2}\int_0^tdt_1\dots\int_0^{t_{k-1}}dt_k\nn\\
%&&
%U_{j}(t,t_1)\Delta^=_{j}\dots U_{j-k+1}(t_{k-1},t_{k})\Delta^=_{j-k+1}\cE^0_{k-2}(0)\nn\\
%&+&
%\int_0^tdt_1\int_0^{t_1}dt_2\dots\int_0^{t_{[j/2]-2}}dt_{[j/2]-1}\nn\\
%U_{j}(t,t_1)\Delta^=_{j}U_{j-2}(t_1,t_2)\Delta^=_{j-2}
%&\dots&
%U_{j-2[j/2]+2}
%(t_{[j/2]-2},t_{[j/2]-1})
%\Delta^=_{j-2[j/2]+2}\nn\\
%&&\cE^0_{j-2[j/2]}(t_{[j/2]-1})\nn
%\eea
%Since, according to the parity of $j$, $j-2[j/2]$ is equal to $0$ or $1$, so is $\cE^0_{j-2[j/2]}(t_{[j/2]-1})$ for all $t_{[j/2]-1}$.
%Therefore we get easily, using \eqref{hyp1},\eqref{hyp2},\eqref{hyp4}, that
%\bea
%\norm{\cE^0_j(t)
%%-U_j(t,0)\cE^0_j(0)
%}
%&\leq&
%%\frac{t^{[j/2]-1}}{([j/2]-1)!}
%%e^{(j+[j/2]-1)|t|}
%\sum_1^{[j/2]-1} \frac{|t|^k}{k!}e^{(j+k)|t|}
%(j^2)^{[j/2]-1}\nn\\
%&\leq&
%e^{|t|}(e^{3|t|}j^2)^{j/2}\nn
%\eea
\be\label{eoj}
\norm{\cE^0_j(t)}\leq
%e^{2j|t|}j
(C'Ae^{C|t|}j^2)^{j/2},
\ee
and thus, by  \eqref{gro}, \eqref{hyp2} and  using $j^\lambda\leq e^{j\lambda/e},\lambda>0$,
\be\label{jodd}
\norm{\int_0^tU_j(t,s)(\Delta_j^+\cE_{j+1}^{0}(s)+\Delta_j^-\cE_{j-1}^{0}(s))ds}_1
\leq
2t(C'Ae^{4/e}e^{(C+1)|t|}j^2)^{j/2},
\ee
and the same argument as the one which leads to \eqref{eoj}, we get, for $j$ odd,
\be\label{e1j}
\norm{\cE^1_j(t)}\leq
%e^{2j|t|}j
(1+2|t|)(C'Ae^{4/e}e^{(C+1)|t|}j^2)^{j/2}.
\ee

For $k>1$ we will estimate $\norm{cE_j^k(t)}_1$ by iterating the second line $M$ times, we will end up with the sum of %less than  $M^3$ 
$3^M$ terms involving the values $\cE^{k-s-t}_{j-2r+s-t}$ for any (r,s,t) such that $M=r+s+t$ with the two constraints $k-s-t\geq 0,\ j-2r+s-t\geq 0$. Using the first constraint we see that
$$
j-2r+s-t\leq j-2r+k\leq j-2(M-k)+k=j-2M+3k.
$$
So that, taking $M=[{(j+3k)}/2]$, the second constraint reduces to $j-2r+s-t= 0$ and the first one to $s+t=k$ since  $\cE_0^k=\delta_{k,0}$.

We easily (and very roughly) estimate, using respectively $M=[{(j+3k)}/2]$, 
\eqref{hyp4}, \eqref{hyp1} and \eqref{hyp2},
\bea
\norm{\cE^k_j(t)}
&\leq&
3^{(j+3k)/2}\frac{|t|^{(j+3k)/2}}{((j+3k)/2)!}e^{3(j+k)|t|/2}((j+k)^2)^{\frac{j+3k}2}\nn
%&\leq&
%(1+3n)^{3(n+1)}e^{3n|t|}j^{3(n+1)}(1+n/j)^j (e^{2|t|}j^2)^{j/2}\nn
%j^{3(n+1)},
\eea
so that, using $(1+k/j)^j\leq e^k$, $j^\lambda\leq e^{j\lambda/e},\lambda>0$ and $n!\geq   n^ne^{-n}$\ \footnote{since 
%$\log$ being increasing we have 
$\log{n!}=\sum\limits_{j=2}^n\log{j}\geq\int_{1}^n\log(x)dx=\big[x\log{x}-x\big]_1^n=n\log{n}-n+1$.}, we get
\be\label{ek}
\norm{\cE^k_j(t)}
\leq
%(1+3n)^{3(n+1)}e^{(3n+1)|t|}(e^{2|t|+3(n+1)/2}j^2)^{j/2}
%(2|t|e^{|t|}k)^{3k/2}
(2|t|e^{|t|+\tfrac53}(3+k))^{3k/2}
(3e^{6k/e}|t|e^{3|t|}j^2)^{j/2}
,\  k>1\nn
%((j+3n)/2)^3e^{(j+3n)|t|}((j+n)^2)^{\frac{j+3n}2}.
\ee
and, 
%since $A\geq 1$,
\be\label{eok}
\norm{\cE^k_j(t)}
\leq
%(1+3n)^{3(n+1)}e^{(3n+1)|t|}(Ae^{2|t|+3(n+1)/2}j^2)^{j/2}
%(2(1+|t|)e^{|t|}k)^{3k/2}
(2(1+|t|)e^{|t|+\tfrac53}(3+k))^{3k/2}
((3e^{\tfrac{6k}e}|t|e^{3|t|}+(1+2|t|)C'Ae^{4/e}e^{(C+1)|t|})j^2)^{j/2}
,\  k\geq 0.\nn
%((j+3n)/2)^3e^{(j+3n)|t|}((j+n)^2)^{\frac{j+3n}2}.
\ee
We conclude by \eqref{hyp2}:

%
%\vskip 5cm
%one find easily that
\be\label{deltaeok}
\norm{
\Delta_j^+(\cE_{j+1}^{k}(s))+\Delta_j^-(\cE_{j-1}^{k}(s)
}\leq 
%%(\frac1N)^{\frac{n+1}2}
%e^{(j+n)|s|}\tfrac{2^n}{n!}(2j^2)^n
%(A_s(j+n)^2)^{\frac{j+n}2}.
%$$
%%$$
%%\leq N^{-(n+1)/2}2^le^{jt}(A\sca(j)^2)^{j/2}e^{n(1+a)}(\sca(j+n)^2)^{n/2}
%%$$
%$$
%\leq 
%%N^{-(n+1)/2}2^le^{jt}e^{n(1+a)}
%%\frac{e^{n|s|}2^n}{n!}
C_k(s)(C'_k(s)j^2)^{j/2}
\ee
%for some $A^n>A$ and $C_n$ (for example, if $\sca(x)=x^{1+\alpha}$, $C_n=\frac{e^{n|s|}2^n}{n!}e^{n^2},\ A^n=e^ne^{|s|}$..
%Indeed
%
%\bea
%&&
%e^{(j)|s|}j^{2n}
%(A_s\sca(j+n)^2)^{\frac{j+n}2}\nn\\
%&&=A_s^{\frac n2}
%(e^{2|s|}A_s\sca(j)^2)^{j/2}\left(\frac{\sca(j+n)}{\sca(j)}\right)^j(\sca(j+n))^nj^{2n}
%\eea
%We have:
%\be\j^{2n}\leq (e^{4n/e})^{j/2}
%\ee
%Let us call $f(y)=\sup\limits_{x\geq 1}\frac{\sca(xy)}{y\sca(x)}$ and $a:=\norm{f'}_\infty$
%\bea
%(\sca(j+n))^{\frac n j}&=&
%e^{\frac n j \log{j\sca(j)}+\frac n j \log(\frac{\sca(j(1+n/j))}{\sca(j)}}\nn\\
%&\leq &e^{bn}e^{\frac{n^2a}{j^2}}
%%\leq e^{bn}e^{an}\nn
%\eea
%where $b=\sup\limits_{y\geq 1}\log{(y\sca(y))}/y$. So
%\be
%(\sca(j+n))^n\leq e^{{an^2}}(e^{2bn})^{j/2}.
%\ee
%Finally
%\bea
%\left(\frac{\sca(j+n)}{\sca(j)}\right)^j
%&\leq&
%((1+\frac n j)f(1+\frac n j))^j\nn\\
%&\leq&
%e^ne^{j\log(1+a\frac n j)}\leq e^{n(1+a)}.\nn
%\eea
%Putting all together we find:
with, after restoring the dependence in $\tfrac{{\norm{V}_{L^\infty}}}\hbar$ by the same argument as  in \cite{PPS}, Section 3, namely a rescaling of the time and the kinetic part of the Hamiltonian,
\be\label{csas}
\left\{
\begin{array}{rcl}
C_k(s)&=&
%2
4e
(2(1+{\tfrac{|s|{\norm{V}_{L^\infty}}}\hbar})e^{{\tfrac{|s|{\norm{V}_{L^\infty}}}\hbar}}k)^{3k/2}\\
&&\times 
%&
(3e^{\tfrac{6k+4}e}{\tfrac{|s|{\norm{V}_{L^\infty}}}\hbar}e^{3{\tfrac{|s|{\norm{V}_{L^\infty}}}\hbar}}+C'Ae^{C{\tfrac{|s|{\norm{V}_{L^\infty}}}\hbar}})^{1/2}
\\
%j(j+1)(1+3n)^{3(n+1)}e^{(3n+1)|t|\tfrac{{\norm{V}_{L^\infty}}}\hbar}
%{(4e^{1+a+|s|})^n}\frac{e^{an^2}}{n!}
%\frac{(4\sqrt{A_s}e^{|s|+1}n)^n}{n!}
%\ \mbox{ and }\ 
C'_k(s)&=&
(3e^{\tfrac{6k}e}{\tfrac{|s|{\norm{V}_{L^\infty}}}\hbar}e^{3{\tfrac{|s|{\norm{V}_{L^\infty}}}\hbar}}+(1+2\tfrac{|s|{\norm{V}_{L^\infty}}}\hbar)C'Ae^{4/e}e^{(C+1){\tfrac{|s|{\norm{V}_{L^\infty}}}\hbar}})e^{6/e}
%Ae^{2|t|\tfrac{{\norm{V}_{L^\infty}}}\hbar+3(n+1)/2}.
%2^{|s|}e^{4n/e}e^{2bn} A_s
%e^{2(1+|s|)+4n/e}A_s.
\end{array}
\right.
\ee
Therefore \eqref{e1} is satisfied and Theorem \ref{main} is proven.

The values of the two constants $D_n(t),D'_n(t)$ in \eqref{dnt} can be expressed out of \eqref{csas} by taking, by Theorem \ref{mainpps}, $C=\sup{(B_1,C_1)},C'=\sup{(B_2,C_2)}$ where $B_1,C_1,B_2,C_1,C_2$ are given in Theorem 2.2. in \cite{PPS}.

\begin{Rmk}\label{crucial}We  see that the  properties \eqref{E0q}-\eqref{endowq}, together with \eqref{normpropag}, are actually the only ones being used in the proof of Theorem \ref{main}.
\end{Rmk}

%\vskip 3cm
\section{Explicit computations of first orders
%, link with previous works and applications
}\label{explicit}

We have
\be\nn
\partial_tU^0_1(t,s)=\frac1{i\hbar}[-\hbar^2\Delta+V_F,U^0_1(t,s)]
+\frac1{i\hbar}[V_{U^0_1(t,s)},F]
\ee
where, in the last term, $V_{U^0_1(t,s)}$ acts on $E_1(s)$ as $V_{U^0_1(t,s)E_1(s)}$.

More generally,
\\
\be\nn
\partial_tU^0_j(t)=\frac1{i\hbar}[-\hbar^2\Delta_{\bR^{jd}}+V_F^{\otimes j},U^0_j(t)]+
%\frac1{i\hbar}[V\star U^0_j(t),F],
P(U^o_j,F)
\ee
where
\be\nn
%V\star E_j
(P(U^o_j,F)E_j)(Z_j)=
\ee
$$
\sum_i\int dx(V(x_i-x)-V(x_i'-x))(U^0_j(t,s)E_j(Z_j^{\neq i},(x,x))F(x_i,x_i'),
$$

 that is
$$
(P(U^0_j,F)E_j)=\suml_{i=1}^j[V\star_i(U^0_j(t,s)E_j),F]_i.
$$

%Formulas \eqref{kee} and \eqref{keeinv} read
%\bea
%\FF^N_j(Z_j)&=&\sum_{k=0}^j\sum_{1\leq i_1,\dots,i_k\leq j}
%\FF(z_{i_1})\dots F(z_{i_k})E_{j-k}(Z_j^{\neq i_1,\dots,i_k})\nn\\
%E_j(Z_j)&=&
%\sum_{k=0}^j\sum_{1\leq i_1,\dots,i_k\leq j}(-1)^k
%\FF(z_{i_1})\dots F(z_{i_k})F^N_{j-k}(Z_j^{\neq i_1,\dots,i_k})\nn
%\eea
%and Corollary \ref{cormain} is reformulated as follows:
%for all $n\geq1$ and $N\geq 4(eA^{2n}_t\sca(j))^2$,
%\bea
%%&&\norm{\FF^N_j(t)-\FF(t)^{\otimes j}-N^{-1}\FF(t)^{\otimes (j-1)}\cE_{j,n}
%%-\sum_{\ell=2}^jN^{-\ell/2}\FF(t)^{\otimes (j-\ell)}\cE_{j,n-\ell}}\nn\\
%%&&
%%\norm{\FF^N_j(t)&-&\sum_{\ell=2}^j\FF(t)^{\otimes (j-\ell)}E_{j-\ell}^{n-\ell}}
%\tr|\FF^N_j(t)&-&
%``\sum_{k=0}^j\sum_{1\leq i_1,\dots,i_k\leq j}
%\FF(z_{i_1})\dots F(z_{i_k})E^n_{j-k}(Z_j^{\neq i_1,\dots,i_k})"
%%\suml_{K\subset \{\un{j}\}}\margetracej{K}{\FF}{\{\un{j}\}}\Phi_{\{\un{j}\}/K}E_{j-|K|}^{n
%%-\frac{(j-|K|)+((-1)^{j-|K|}-1)/2}2
%%}
%|\nn\\
%&\leq&
%2C_{2n}(t)N^{-n-1}.\nn
%\eea
%%for $N\geq 4(eA^{2n}_t\sca(j))^2$.
%where $``G(Z_j)"$ is meant for the operator of integral kernel $G(Z_j)$.

%Moreover
Finally
\be\label{qE20}\nn
\cE^0_2(t)(Z_2)=\int_0^t\int_{\bR^{2d}}dsdZ'_2U_2(t,s)(Z_2,Z'_2)V(x'_1-x'_2)\FF(s)(z'_1)\FF(s)(z'_2) dsdZ'_2
\ee
and
\bea\label{qE11}\nn
\cE^1_1(t)&=&\int_0^tU_1(t,s)Q(\FF,\FF) ds\nn\\
&&+(1-\tfrac1N)\int_0^t\int_0^s
U_1(t,s)\tr^2[VU_2(s,u)V\FF(u)\otimes\FF(u)]dsdu\nn\\
\eea
%\subsection{Link with \cite{Spohnboltz}}
%In \cite{Spohnboltz} is used another definition of the error term which consists in replacing $\FF(t)$, the solution of the mean field equation, by $\FF^N_1(t)$:
%\be\label{errorspohn}
%\widetilde E_j(Z_j)=
%\sum_{k=0}^j\sum_{1\leq i_1,\dots,i_k\leq j}(-1)^k
%\FF^N_1(z_{i_1})\dots F^N_1(z_{i_k})F^N_{j-k}(Z_j^{\neq i_1,\dots,i_k})
%\ee
%We get that $\widetilde E^1(t)=0$ ($E^1(t)=\FF(t)-\FF^N_1(t)$). Therefore the subleading order is given by $\widetilde E_2(t)$:
%\bea
%\FF^N_j(t)&=&(\FF_1^N(t))^{\otimes j}\nn\\
%&+&\sum_{1\leq i_1<i_2<\dots<i_{j-2}\leq j}
%\margetracej{\{i_1,\dots,i_{j-2}\}}{\FF}{\{\un{j}\}}\Phi_{\{\un{j}\}/\{i_1,\dots,i_{j-2}\}}\widetilde E_2(t)\nn\\
%&+& O(N^{-2}).\nn
%\eea
%therefore
%\bea
%\widetilde E_2-E_2&=&
%(\FF_1^N)^{\otimes 2}-2(\FF^N_1)^{\otimes 2}+\FF^N_2\nn\\
%&-&(\FF^{\otimes 2}-\FF\otimes\FF^N_1-\FF^N_1\otimes\FF+\FF^N_2)\nn\\
%&=&- E_1^{\otimes 2}=O(N^{-2}).\nn
%\eea
%since $\FF^N_1=\FF+E_1$.
\section{The Kac and ``soft spheres" models}\label{kac}
In this section we consider the two following classes of mean field models (see \cite{PPS} for details).

\medskip
\noindent
$\bullet$ {\it Kac model.} In this model, the $N$-particle system evolves according to a stochastic process. To  each particle $i$, we associate a velocity $v_i \in \R^3$.
\newcommand{\calv}{\cal V}
The vector $\calv_N=\{v_1, \cdots, v_N\}$ changes by means of two-body collisions at random times, with random scattering angle.
The probability density $F^N(\calv_N,t)$ evolves according to the forward Kolmogorov equation 
\be\label{kolmo}
\pa_t F^N=\frac 1N \sum_{i<j} \int d\omega B(\omega; v_i-v_j)\{F^N(\calv_N^{i,j})-F^N(\calv_N)\}\;,
\ee
where $\calv_N^{i,j}=\{v_1, \cdots,v_{i-1}, v_i', v_{i+1}, \cdots, v_{j-1},v_j', v_{j+1}, \cdots, v_N\}$ and the pair $v_i' ,v_j'$ gives the outgoing velocities after a collision with scattering (unit) vector $\o$ and incoming velocities $v_i, v_j $. $\tfrac{B(\omega; v_i-v_j)}{|v_i-v_j|}$ is the differential cross-section of the two-body process.
The resulting mean-field kinetic equation is the homogeneous Boltzmann equation
\be\label{eqboltzz}
\pa_t F(v) = \int dv_1 \int d\omega B(\omega; v-v_1) \{ F(v') F(v_1') -F(v) F(v_1) \}\;.
\ee
%Such a model has been introduced by Kac \cite{Kac,Kac2} and has been largely investigated over recent times, see e.g.\,\cite {MM}.
%Very similar stochastic systems including space variables and (space-)delocalized collisions are frequently used to justify numerical schemes 
%\cite{PWZ,RW}. We will not mention them explicitly although they could be included in our analysis.

\smallskip
\noindent
$\bullet$  {\it `Soft spheres' model.} A slightly more realistic variant, taking into account
the positions of particles $X_N= \{x_1, \cdots, x_N\} \in \R^{3N}$ and relative transport,
was introduced by Cercignani \cite{Cer} and further investigated in \cite{LP}.
The  probability density $F^N(X_N, V_N,t)$ evolves according to the equation
\bea
\pa_t F^N + \suml_{i=1}^N v_i \cdot \nabla_{x_i} F^N &=&
\frac 1N \suml_{i<j} \, h\left(|x_i-x_j|\right) B\left(\frac{x_i-x_j}{|x_i-x_j|}; v_i-v_j\right)
\nn \\
&&
\times\{F^N(X_N ,V_N^{i,j})-F^N(X_N,V_N)\}\;. \label{kolomosoft}
\eea
Here $h:\R^+ \to \R^+$ is a positive function with compact support. 
Now a pair of particles collides at a random distance with rate 
modulated by $h$. 
The associated mean-field kinetic equation is the Povzner equation
\bea\label{eqpovzner}
\pa_t F(x,v) + v\cdot \nabla_{x}F (x,v) &=&\int dv_1 \int dx_1\, h(|x-x_1|) B\left(\frac{x-x_1}{|x-x_1|}; v-v_1\right)  \nn \\
&&\times\{ F(x, v') F(x_1, v_1') -F(x, v) F(x_1, v_1) \}, \nn 
\eea
which can be seen as an $h-$mollification of the inhomogeneous Boltzmann equation (formally obtained when $h$
converges to a Dirac mass at the origin). 
 Both classes have be treated in \cite{PPS} and Theorem \ref{mainpps} apply to them, in the following sense.
 
 The underlying space $\lon$ is now $L^1(\bR^d,dv)$ (resp. $L^1(\bR^{2d},dxdv)$)) for the Kac model (resp. soft spheres) both  endowed with the $L^1$ norms $\norm{\cdot}_1$. For $F^N\in\lone{N}$, $F^N_j\in\lone{j}$ is defined by $$F^N_j(Z_j)=\int_\Omega F^N(z_1,\dots,z_j,z_{J+1},\dots,z_N)dz_{j+1}\dots dz_N$$ for $Z_n=(z_1,\dots,z_n),n=\unN$ with $z_i=v_i\in\bR^d,\Omega=\bR^{(N-j)d}$ (resp. $z_i=(x_i,v_i)\in\bR^{2d},\Omega=\bR^{2(N-j)d}$) for the Kac (resp. soft spheres) model. 
 
 In both cases $E_j(t)$ is defined by \eqref{deferrorq}, inverted by \eqref{invdeferrorq}, and it was proven in \cite{PPS} that Theorem \ref{mainpps} holds true verbatim in both cases.
 
 Stating now the dynamics driven by \eqref{kolmo} and \eqref{kolomosoft} under the form \eqref{Nsop} with $K^N=0$ (resp. $K^N=-\sum\limits_{i=\unN}v_i\partial_{x_i}$) for the Kac (resp. soft spheres) model and $V^N$ given by the right hand-sides of \eqref{kolmo},\eqref{kolomosoft} respectively, one sees immediately  that the proofs contained in Sections \ref{recur},\ref{esti} remain valid after an elementary redefinition of the operators $D_j,D_j^{-1},D_j^{-2}$ in \eqref{ddefj}-\eqref{defd=} 
consisting in removing the bottom and overhead straight lines in the right hand sides and, by a slight abuse of notation,  identifying functions with their evaluations. The convention \eqref{E0q} remains verbatim the same, together with the estimates
\be\label{endow}
\| D_j \| , \| D^1_j\| \leq j  \mbox{ and }  \| D_j^{-1}\|,
\| D_j^{-2}\|, \|D^{-1}_1(E_0)\|,  \|D^{-2}_2(E_0)\|
 \leq \frac {j^2} N.
\ee

 Therefore, by Remark \ref{crucial}, the statements contained in Theorem \ref{main} and Corollary  \ref{corcormain} hold true, in both cases, verbatim. Moreover defining $F^{N,n}_j$  by \eqref{FNnj} in both cases, Corollary \ref{cormain} reads now as follows
 
 \begin{Cor}\label{cormaink}[Kac case]
Let $\FF^N(t)$ the solution of the $N$ body system \eqref{kolmo} (resp. \ref{kolomosoft})  with initial datum $\FF^N(0)=\FF^{\otimes N}$, $0<\FF\in L^1(\bR^{d})), \int\limits_{\bR^d}f(v)dv=1$ (resp, $0<\FF\in L^1(\bR^{2d})), \int\limits_{\bR^{2d}}f(x,v)dxdv=1$), and $\FF(t)$ the solution of the homogeneous Boltzmann equation \eqref{eqboltzz} (resp. the Povzner equation\eqref{eqpovzner}) with initial datum $\FF$.

Then, in both cases, for all $n\geq1$ and $N\geq 4(eA^{2n}_tj)^2$,
\bea
\norm{\FF^N_j(t)-
F^{N,n}_j(t)
}_1
\leq
N^{-n-\frac12}\ 
\tfrac{2tC_{2n}(t)eA^{2n}_tj}{\sqrt N}.\nn
\eea
%for $N\geq 4(eA^{2n}_tj)^2$.
\end{Cor}
\begin{appendix}
\section{The asbtract model}\label{abstract}
\subsection{The model}\label{model}
We will show in this section that the main results of \cite{PPS} and of Section   \ref{into} of the present paper remain true in the ``abstract`" mean field formalism for a dynamics of $N$ particles that we will describe now. The present formalism contains the   abstract formalism developed in \cite{PPS}, without requiring a space of states endowed with a multiplicative structure.
\subsubsection*{States of the particle system and evolution equations}\label{state}

Let $\lon$ be a vector space on the complex numbers. 
%For any $n\leq N$ we set
% %We define the Banach space $\lonN$ as
%\be\label{lone}
%\lon^n:=\lon^{\otimes n}.
%\ee
We  suppose the family of (algebraic) tensor products $\{\lone{n}, n=1,\dots,N\}$ equipped with a family of norms $\norm{\cdot}^n$  satisfying    assumption  $\ref{a}$ below. the $N$-body dynamics will be driven on $\lone{N}$ by a one- and two- body  interaction  satisfying assumption $\ref{b}$
%, In order to derive a BBGKY type hierarchy we will also assume $\ref{c}$ 
and the mean field limit equation will be supposed to satisfy assumption $\ref{c}$. 

Assumptions $\ref{a}-\ref{c}$ below will be   followed by their incarnations in the {\bf K}(ac), {\bf S}(oft spheres) and {\bf Q}(uantum) models.

By convention we denote $\lone{0}:=\bC,\ \norm{z}^0=|z|$ and we  denote by $\lon^{\hat{\otimes}n}$ the  completion of $\lone{n}$ with respect to the norm $\norm{\cdot}^n$.

\textit{\ksq $\lone{n}$ is $L^1(\bR^d,dv), L^1(\bR^{2d},dxdv)$ and $\cL^1(L^2(\bR^d)$, the space of trace class operators on $L^2(\bR^d)$, with their  associated norms.}

\vskip 1cm

\begin{enumerate}[label=(\bf\Alph*)]
\item\label{a}   There exists a family of subsets $\lon^{\hat\otimes n}_+$ of $\lon^{\hat\otimes n}, n=\unN$, of positive elements $F$ denoted by 
 $F >0$  stable by addition, multiplication by positive reals and tensor product and  there exists a linear function  $\tr:\ \lon\to\bC$, called trace.
% \newpage
 For every $1\leq k,n\leq N$ and $1\leq i\leq j\leq n\leq N$, let $\tr^k_n$ and $\sigma^n_{i,j}$ be  the  two mapping defined  by\footnote{The fact that the second and fourth lines of \eqref{tr}  define a mapping on the whole tensor space $\lone{n}$ results easily from the definition of tensors products through the so-called universal property  \cite{lang}. Indeed, let $\varphi_n$ be the natural embedding $\lon^{\times n}\to\lone{n}, (v_1,\dots,v_n)\mapsto v_1\otimes\dots\otimes v_n$, and let $h$ be any mapping  
$\lon^{\times n}\to \lon^{\times n'}$, then the universal property of tensor products says that  there is a unique map $\tilde h: \lone{n}\to\lone{n'}$ such that $\tilde h\circ \varphi_n=\varphi_{n'}\circ h$. 
Taking $n'=n-1$, $h(v_1,\dots,v_k,\dots,v_n)=(\Tr(v_k)v_1,\dots,v_{k-1},v_{k+1},\dots,v_n)$ for $\tr^k_n$, and 
$n'=n$, 
$h(v_1,\dots,v_i,\dots,v_j,\dots,v_n)=
(v_1,\dots,v_j,\dots,v_i,\dots,v_n)$ for $\sigma^n_{i,j}$ 
give the desired 
extensions.}
\newpage
\be\label{tr}
\left\{
\begin{array}{lccl}
\marg{k}_n: &\lone{n}&\to&\lone{n-1}\\
&\underset{\i=1}{\overset{n}{\otimes}}v_i&\mapsto&
\tr(v_k)\underset{\substack{\i=1\\i\neq k}}{\overset{n}{\otimes}}v_i,\\
%\end{array}
%\ee
%
%\be\label{sigmanij}
%\begin{array}{lccll}
\sigma^n_{i,j}: &\lone{n}&\to&\lone{n}\\
&\underset{\i=1}{\overset{n}{\otimes}}v_i&\mapsto&
\underset{\i=1}{\overset{n}{\otimes}}v'_i,\ \ \ 
v'_k=v_k,i\neq k\neq j\ ;\ v'_i=v_j,v'_j=v_i
.
\end{array}
\right.
\ee
We will suppose that $\tr^k_N$ and $\sigma^n_{i,j}$, $i,j,k\leq n\leq N$, 
%are positivity and norm preserving, that is:
satisfy, for any $F\in\lone{n}$,
\be\label{posnorm}
\left\{
\begin{array}{l}
\tr^k_N(F), \sigma^n_{i,j}(F)>0,\ \norm{\tr^k_n( F)}_{n-1}=\norm{F}_n\mbox{ when }F>0\\
\norm{\sigma^n_{i,j}(F)}^n=\norm{F}^n\\
\norm{\tr^k_n( F)}^{n-1}\leq\norm{F}^n
\end{array}
\right.
\ee
In particular one has that $\norm{F}^n=\tr^n\dots\tr^1F$ when $F>0$ and $|\tr^n\dots\tr^1F|\leq\norm{F}^n$ in general.

%\vskip 0.3cm\item\label{c} each $\tr^k_N$ is positivity and norm preserving, that is:
%\be\label{posnorm}
%\left\{
%\begin{array}{l}
%\tr^k_N(F)>0\mbox{ when }F>0\\
%\norm{\tr^k_n F}_{n-1}=\norm{F}_n
%\end{array}
%\right.
%\ee
\noindent Note that \eqref{posnorm} allows to extend  $\tr^k_n$ and $\sigma^n_{i,j}$ to $\lon^{\hat\otimes n}$ by continuity. We will use  the same notation for these extensions.

{\it \ksq $\tr^k$ is $\int_{\bR^d} \cdot dv_k,\ \int_{\bR^{2d}} \cdot dx_kdv_k$ as indicated in Section \ref{kac}, and the partial traces defined in Section \ref{quantum}.
%, namely, defining $A$ through its integral kernel\\ $a(x_1,\dots,x_n,x'_1,\dots,x'_n)$, the integral kernel of $\tr^kA$ is defined as \\$\int_{\bR^{2d}}a(x_1,\dots,x_n,x'_1,\dots,x_k,\dots,x'_n)dx_k$, respectively. Moreover, t
The action of $\sigma^n_{i,j}$ consists obviously in exchanging the variables $v_i$ and $v_j$, $(x_i,v_i)$ and $(x_j,v_j)$ and $(x_i,x'_i)$ and $(x_j,x'_j)$, (in the integral kernel), 
respectively. Finally \eqref{posnorm} is satisfied in the three cases.}
\end{enumerate}

 From now on and when no confusion is possible, we will identify $\lone{n}$ with its completion $\lon^{\hat{\otimes}}$ and  we will denote $\tr^k_N=\tr^k$ (note also that $\tr=\tr^1_1=\tr^1$), $\sigma^N_{i,j}=\sigma_{i,j}$ and $\tr(=\tr_n)=\tr^n_n\tr^{n-1}_n\dots\tr^1_n$.
Moreover, with a slight abuse of notation,  we will denote
\be\label{normsa}
\left\{
\begin{array}{l}
\norm{\cdot}_1= \norm{\cdot}^n,\ \forall n=\unN\\
\norm{\cdot} \  \mbox{ the operator norm on any }\cL(\lone{i},\lone{j}),\ \forall i,j=\unN
%\mbox{ or }\lone{j+1}\to\lone{j}.
\end{array}\right.
\ee
(here $\cL(\lone{i},\lone{j})$ is the set of bounded operators form $\lone{i}$ to $\lone{j}$). 

 %\noindent 
 We call \textit{symmetric} any element  of $\lone{n}$ invariant by the action of $\sigma^n_{i,j},\ i,j\leq~ n$.

%\vskip 1cm
We call {\em state of the $N-$particle system}  an element  of
\be
\label{cone}
%(\lone{n})^+_1
\cD_N= \{ F \in \lone{n} \,\,|\,\, F >0, \quad \| F\|=1 \quad \text {and $F$ is symmetric}\}.
\ee
%Then the following identity will play an important role in the sequel.
%A {\em state of the $N-$particle system} is, by definition, an element in $(\lone{N})^+_1$ symmetric
%in the particle labels, meaning that
%\bea\label{symbas}
%\FF^N&=&\sum_{\ell_1,\dots,\ell_N}
%c_{\ell_1,\dots,\ell_N}
%u^{\ell_1}\otimes\dots\otimes u^{\ell_N}\nn\\
%&=&
%\sum_{\ell_1,\dots,\ell_N}
%c_{\ell_1,\dots,\ell_N}
%u^{\ell_{\sigma_1}}\otimes\dots\otimes u^{\ell_{\sigma_N}}
%\eea
%for any {\color {blue} basis} {\color {red} finite set of independent  vector}  $\{u^\ell\}_{\ell \in \bN}$  of $\lon$ and
%any permutation $\sigma = (\sigma_1,\cdots,\sigma_N)$ of the set $\{1,\dots, N\}$. 
%

%\medskip

For $j=0,\dots,N$, the $j$-particle marginal of $F^N \in (\lone{N})^+_1$ is defined as the 
 the partial trace of order $N-j$ of $F^N $, that is
\be
\label
{jmarg}
F^N_j =\tr^N \tr^{N-1} \cdots \tr^{j+1} F^N, F_N^N:=F^N.
\ee
Note that $F^N_j\in \lone{j}$ ($F^N_0=1\in\lone{0}:=\bC$) and $F^N_j >0, \norm{F^N_j}^j=\norm{F^N}^N$ since $\tr$ is positivity and norm  preserving, and 
obviously $\FF^N_j$ is symmetric as $\FF^N$. That is to say: $$F^N_j\in\cD_j.$$

\vskip 1cm
\begin{enumerate}[label=(\bf\Alph*),resume]
\item\label{b} The evolution of a state $\FF^N$ in $\lone{N}$ is supposed to be given by the $N-$particle dynamics associated to a two-body interaction:
\be\label{Neq}
\frac{d}{dt}\FF^N=(K^N+V^N)\FF^N,
\ee
where the operators on the right hand side are constructed as follows.
 \be\label{defV}
K^N=\sum_{i=1}^N\mathbb I_\lon^{\otimes(i-1)}\otimes K\otimes\mathbb I_\lon^{\otimes(N-i)}\;
\ee
and 
\be\label{defV}
V^N=\frac1N\sum_{1\leq i< j\leq N}V_{i,j},\ \ \ V_{i,j}:=\sigma^N_{1,i}\sigma^N_{2,j}V
\otimes  \mathbb I_{\lone{(N-2)}}\sigma^N_{1,i}\sigma^N_{2,j}
\ee
for a (possibly unbounded) operator $K$ acting on $\lon$ and a bounded two-body (potential) operator $V$ acting on $\lone{2}$.

 We assume furthermore that $K$ is the generator of a strongly continuous, isometric, positivity preserving semigroup (in $\lon $)
\be
\label{K+}
e^{Kt} F >0 \quad \text {if} \quad F>0\;; \qquad  \|e^{tK}\|=1\;.
\ee
and  $K^N +V^N $ is the generator of a strongly continuous, isometric, positivity preserving semigroup (in $\lone{N}$)
\be
\label{KV+}
e^{(K^N +V^N) t} F^N >0 \quad \text {if} \quad F^N>0\;; \qquad  \| e^{t(K^N +V^N) } \|=1\;.
\ee
  Finally, for any $F\in\lon$, $F^N \in \lone{N}$ and $i,r >j$,  we assume
\be\label{trackf}
\tr(KF)=0\mbox{ and }
\tr^{j,N} (V_{i,r} F^N)=0\;.
\ee
This last property is necessary to deduce the forthcoming hierarchy.

{\it \ksq the ingredients in \eqref{Neq} are given 
in Sections  \ref{kac} and \ref{quantum}, where
%in the table bellow and  
\eqref{K+}-\eqref{trackf} are 
%easily 
shown to be satisfied.}
\end{enumerate}

\medskip
Note the symmetry property of the equation \eqref{Neq} induced by the definition of $V^N$: if the initial condition $\FF^N_0$ for \eqref{Neq} is symmetric, then $\FF^N (t) $ is still symmetric.

\subsubsection*{Hierarchies}

%Applying successively $\marg{N},\marg{N-1}\dots$ to \eqref{Neq} and using \eqref{trackf}, we get 

%We prove in Appendix \ref{absbbgky} that t
The family of $j$-marginals, $j=\unN$, are solutions of the BBGKY hierarchy of equations
%\be\label{bbgky}
%\FF
%For a sequence of marginals $\{ f^N_J \}$ we consider the evolution:
\be\label{eqhiera}
\pa_t \FF^N_j = \left(K^j+\frac {T_j}{N}\right) \FF^N_j+ \frac {(N-j)} N C_{j+1} \FF^N_{j+1}
\ee
%where $\a(j,N)=\frac {(N-j)} N $.
where:
\be
K^j=\sum_{i=1}^j\mathbb I_\lon^{\otimes(i-1)}\otimes K\otimes\mathbb I_\lon^{\otimes(j-i)},
\ee
\be
T_j=\sum_{1\leq i < r \leq j} T_{i,r}\ \ \ \ \ \ \mbox{ with }\  T_{i,r}=V_{ir}
\ee
and
\be
\label{C}
C_{j+1}\FF^N_{j+1}=\marg{j+1}\left(\sum_{i\leq j}V_{i,j+1}\FF^N_{j+1}\right)=\sum_{i=1} ^jC_{i,j+1}\FF^N_{j+1},
\ee
%with
\be
C_{i,j+1}: \lon^{\otimes (j+1)}\to\lon^{\otimes j},\ \ \ C_{i,j+1}\FF^N_{j+1}=\marg{j+1}\left(V_{i,j+1}\FF^N_{j+1}\right)\;,
\label{eq:Ci}
\ee
%\be
%C_{i,j+1}: \lon^{\otimes (j+1)}\to\lon^{\otimes j}.
%\ee
Indeed, thanks to \eqref{trackf} we get easily by applying $\tr^{j,N}$ on \eqref{Neq} that
$$
\frac d{dt}F^N_j=(K^j+\frac{T_j}N)F^N_j
+\frac1N\tr^{j,N}(\sumls{1\leq i\leq j< k\leq N} V_{i,k}F^N)
$$
By symmetry of $F^N$ and $V_{i,k}$ we get $\tr^{j,N}(V_{i,k}F^N)=\tr^{j+1}(V_{i,j=1}F^N_{j+1})$ for all $k>j$ and  \eqref{eqhiera} follows.
%\vskip 1cm
%\noindent For the sequel, we will need a (bit) more sophisticated definition of marginals and other objects. Though the definitions might look a bit heavy, let us point out from the beginning that the proof will use very few fact concerning them. 

Note that, thanks to  the assumption \eqref{posnorm} and for all $ i \leq j=\unN$,
%Thanks to the different compatibility relations expressed earlier and the definition \eqref{defV}, one immediately get that, $\forall i.j=\unN$,
\be\label{normtc}
\norm{T_i}\leq j^2\norm{V},
%\leq j\norm{V} \mbox{ and }
\mbox{ and }\norm{C_{i,j+1}}\leq j\norm{V}
\ee
(meant for ($\norm{T_i}_{\lon^{\otimes i}\to\lon^{\otimes i}},\ \norm{C_{i,j+1}}_{\lon^{\otimes (j+1)}\to\lon^{\otimes j}},\ \norm{V}_{\lon^{\otimes 2}\to\lon^{\otimes 2}}$ using \eqref{normsa}).

\medskip
\vskip 1cm
{ We introduce the non-linear mapping $Q(\FF,\FF)$, $Q: \lon \times \lon \to \lon$ by the formula
\be\label{nlc}
Q(\FF,\FF)=\marg{2}(V_{1,2}(\FF\otimes\FF))
\ee
and the nonlinear mean field equation on $\lon$
\be\label{mfea}
\pa_t\FF=K\FF+Q(\FF,\FF),\ F(0)\geq 0,\ \norm{F(0)}_1=1.
\ee
Eq.\,\eqref{mfea} is the Boltzmann, Povzner or Hartree equation according to the specifications
established in the table above. In full generality we will assume
\begin{enumerate}[label=(\bf\Alph*), resume]
\vskip 0.3cm\item\label{c}
\eqref{mfea} has  for all time a unique solution $F(t)>0$ and $\norm{F(t)}
%=\norm{F(0)}
=1$. 

{\it \ksq \ref{c} is true by standard perturbations methods.}
\end{enumerate}
\vskip 1cm

\subsubsection*{Correlation error.}
To introduce the correlation errors, we need to extend slightly the above structure.
%Though the definitions might look heavy, the proof will use a very few facts. 

For any  subset $J\subset\{\unN\}$ we first define
\be\label{lj}
\lone{J}_N:=\overset{N}{\underset{i=1}{\otimes}}\lone{\chi_J(i)},
\ee
where $\chi_J$ is the characteristic function of $J$ and $\lone{0}=\bC$. 

Then we introduce $\lone{J}$, the subspace of $\lone{J}_N$ formed by vectors of
the form 
$
\overset{N}{\underset{i=1}{\otimes}}v_i
$ where $v_i=1\in\bC$ for $i\notin J$ and $v_i\in\lon$ for $i\in J$. Note that $\lone{J}$ is sent to $\lone{|J|}$ by the mapping $$\Pi: \ 
\overset{N}{\underset{i=1}{\otimes}}v_i
\in\lone{J}\mapsto 
\underset{i\in J}{\otimes}v_i\in\lone{|J|}.$$
We define a norm on $\lone{J}$ by
\be\nn
\norm{\cdot}_{\lone{J}}=\norm{\Pi(\cdot)}_
%{\lone{|J|}}
1.
\ee

%\noindent We can now define, for any subset $J\subset\{\unN\}$, the $J$-marginal of $\FF^N$ as
%\be\label{defJmarg}
%\FF^N_J:=
%\prod_{j\in J}\tr^{j}\FF^N.
%\ee

For $\FF\in\lon$ and $K\subset J\subset \{\unN\}$ we introduce the linear operator
 $\margetracej{{K}}{\FF}{J}$, defined through its action on  factorized elements as
\bea
\margetracej{{K}}{\FF}{J}:\lone{J/K}&\to&\lone{J}\nonumber\\
\label{mgtr} 
\overset{N}{\underset{i=1}{\otimes}} v_i
&\mapsto&
\overset{N}{\underset{i=1}{\otimes}} a_i,
\eea
where $\left\{
\begin{array}{ccccl}
a_s&=&1\in\bC&\mbox{ if }& s\notin J\\
a_s&=&\FF&\mbox{ if }& s\in K\\
a_{s}&=&v_s&\mbox{ if }& s\in J/K
\end{array}
\right.$.

%and for $K\subset\{\un{j}\}$, 
%\be\label{mgtr}
%\margetrace{K}{\FF}:=\prod_{k\in K}\margtrace{{k}}{\FF}:\ \lone{j-|K|}\to\lone{j}.
%\ee
%\bcrr{We will suppose
%\be\label{compatrnorm3}
%\norm{\margtrace{k}{\FF}\FF^n}\leq
%\norm{\FF}\norm{\marg{k}\FF^n}\hskip 4cm\bcrr{\textbf{(H3)}}
%\ee
\vskip 0.5cm
Note that, for $K, K'\subset J,\ K\cap K'=\emptyset$, we have the  composition
\be\label{KK'}
\margetracej{K}{\FF}{J}\margetracej{K'}{\FF}{J/K}=\margetracej{(K\cup K')}{\FF}{J}=\margetracej{K'}{\FF}{J}\margetracej{K}{\FF}{J/K'}
\ee
and more generally, for all $\FF,G$,
\be\label{inter}
\margetracej{K}{\FF}{J}\margetracej{K'}{G}{J/K}=
\margetracej{K'}{G}{J}\margetracej{K}{\FF}{J/K'}.
\ee
%When no confusion arises, we shall drop the lower indices and just write $\margetrace{K}{\FF}$.
%\vskip 1cm
%\medskip
%\vskip 1cm

For any 
subset $J\subset\{\unN\}$,
%$j=\unN$,
 we define the  
 %{\color{blue} ``correlation error coefficients" } 
%{ ``kinetic error " }
{\em correlation error} by
\be\label{deferror}
E_J=\sum_{K \subset J} (-1)^{|K|} 
%f^{\otimes K}ÔøΩf^N_{J/K}
%\big(\prod_{\substack{k\in K}}\margtrace{k}{\FF}\big)
\margetracej{K}{\FF}{J}
\FF^N_{
J
/K
%j-|K|
}
\ee
where 
%$  f^{\otimes K}=\prod_{i\in K} f_i$ and $f_i=f(v_i)$ is a distribution of single particle ($v$ the one-particle variable)
%which satisfies
%\be
%\pa_t f = Q(f,f)
%\ee
%where $Q$ is a bilinear operator connected to $C$ by the formula
%\be
%C_{i,j+1} f ^{\otimes  (j+1) } (v_1 \cdots v_{j+1})= \prod_{r \neq i} f(v_r) Q(f,f) (v_i)
%\ee
%so that
%\be
%\pa_t f^{\otimes j}  = C_{j+1} f^{\otimes (j+1)}.
%\ee
$\FF$ solves \eqref{mfea}, the operator $\margetracej{K}{\FF}{J}$ is defined by \eqref{mgtr} and
%, with the convention $\lon^0:=\bC$, 
$\FF^N_L\in \lone {L}$ is defined through its decomposition on factorized states. Namely if
$$
\FF^N =\sum_{\ell_1,\dots,\ell_N}
c_{\ell_1,\dots,\ell_N}
v_{\ell_{1}}\otimes\dots\otimes v_{\ell_{N}},
$$
then
$$
\FF^N_L =\sum_{\ell_1,\dots,\ell_N}
c_{\ell_1,\dots,\ell_N}
a_{\ell_{1}}\otimes\dots\otimes a_{\ell_{N}},
$$
%where
%\be\label{margj}
%(\overset{N}{\underset{i=1}{\otimes}}v_i)_L=
%\overset{N}{\underset{i=1}{\otimes}}a_i
%\FF^N_L:=\prod_{\substack{j\in \{\unN\}/L}}\marg{j}\FF^N\in\lon^{|L|}
%%\ \ \ (
%%,\ \lon\neq\emptyset;\ \ \ \FF^N_\emptyset:=1\in
%%\lon^0:=\bC).
%%\FF^N_\ell:=\prod_{\substack{j\in \{\ell+1,\dots,N\}}}\marg{j}\FF^N\in\lon^{\ell}\ \ \ (
%%%,\ \lon\neq\emptyset;\ \ \ \FF^N_\emptyset:=1\in
%%\lon^0:=\bC).
%\ee
where $\left\{
\begin{array}{cclcl}
a_s&=&\tr (v_s)\in\bC&\mbox{ if }& s\notin L\\
a_s&=&v_s&\mbox{ if }& s\in L
\end{array}
\right.$.

\noindent
The link between the definition of $\FF^N_L$ and the definition of the marginals $\FF^N_j$ given in \eqref{jmarg} is the following:
\be\label{jmargj}
\FF^N_{\{1,\dots,\ell\}}=\FF^N_\ell\otimes(1)^{\otimes(N-\ell)}\in\lone{ \ell}\otimes(\lone{0})^{\otimes (N-\ell)}.
\ee

The  formula inverse to \eqref{deferror} reads
\be\label{invdeferror}
\FF^N_J=\sum_{K \subset J} 
%\prod_{k\in K}\margtrace{k}{\FF} 
\margetracej{K}{\FF}{J}
E_{J/K
%/K
}.
\ee
Note that the contribution in the right hand side of \eqref {invdeferror} %and \eqref{deferror} 
corresponding to $K=J$ and $K=\emptyset$
are $F^{\otimes |J|} $ and $E_J$  respectively.
To prove \eqref{invdeferror}, we plug \eqref{deferror} in the r.h.s. of \eqref{invdeferror} and we use \eqref{KK'}:
%we get
%\bea
%\sum_{K \subset J} \prod_{k\in K}\margtrace{k}{\FF} E_{J}&=&
%\sum_{K \subset J} \prod_{k\in K}\margtrace{k}{\FF}\nonumber\\
%&\times&
%\sum_{K' \subset J/K} (-1)^{|K'|} 
%\prod_{\substack{k'\in K'\\ j'\in \{\unN\}/
%%(
%J
%%/K/K'
%%)
%}}\margtrace{k'}{\FF}\marg{j'}\FF^N\nonumber\\
%%\mbox{(writing } 
%L:=K\cup K')&=&
%\sum_{L \subset J} \prod_{\substack{l\in L\\j\in\{\unN\}/
%%( 
%J
%%/L)
%}}\margtrace{l}{\FF}\marg{j}\FF^N\sum_{K' \subset L}(-1)^{|K'|}
%\nonumber\\
%&=&
%\prod_{j\in\{\unN\}/J}\marg{j}\FF^N=\FF^N_J\nonumber
%\eea
\bea
\sum_{K \subset J} \margetracej{K}{\FF}{J} E_{J/K}&=&
\sum_{K \subset J} \margetracej{K}{\FF}{J}
\big[
\sum_{K' \subset J/K} (-1)^{|K'|} 
\margetracej{K'}{\FF}{J/K}\FF^N_{(J/K)/ K'}\big]\nonumber\\
&=&\sum_{K\cup K' \subset J} 
\sumls{K \subset J\\K'\cap K=\emptyset}
(-1)^{|K'|} 
\margetracej{K}{\FF}{J} 
 \margetracej{K'}{\FF}{J/K}
 \FF^N_{J/(K\cup K')}
\nonumber\\
%\mbox{(writing } 
%(L:=K\cup K')
&=&
\sum_{L \subset J}(\sum_{K' \subset L}(-1)^{|K'|})\margetracej{L}{\FF}{J}\FF^N_{J/L}
%\big[
=\FF^N_J\nonumber
\eea
since $\sum\limits_{K' \subset L}(-1)^{|K'|}=\sum\limits_{k'=0}^{|L|}\binom{|L|}{k'}(-1)^{|K'}=0^{|L|}=0$ if $L\neq\emptyset$, and $=1$ if $L=\emptyset$ (since  $\sum\limits_{K' \subset \emptyset}(-1)^{|K'|}=(-1)^0=1$).

 %Through this paper w
% We will use the notation
%\be\label{defkro}
%\delta_{K,K'}=1\mbox{ if }K=K'\mbox{ and }0\mbox{ if }K\neq K'.
%\ee
%\vskip 1cm

%Thanks to the symmetry property of $F^N$, one recovers   more generally $\FF^N_{|J|}$ out of $\FF^N_J$  for any $J\subset\{\unN\}$, by using the identification $\lone{0}\otimes \lon\sim\lon\otimes \lone{0}\sim \lon$ and the mapping:
%\bea
%\overset{N}{\underset{i=1}{\otimes}}a_i\in\lone{J} 
%&\mapsto&
% {\underset{j\in J}{\otimes}}v_j\in\lone{|J|}\nn\\
% \FF^N_J&\mapsto&\FF^N_{|J|}\label{mnbv}
% \eea
%where the $a_s$ are the ones of \eqref{margj}.
%In fact in \eqref{margj} $\FF^N_{J/N}$ is considered in an obvious way as an element of $\lone{J/K}$ as defined in \eqref{lj}.

%In order to invert \eqref{mnbv}, o
One notices that since $\FF^N_j$ is the marginal of some $\FF^N$ which decomposes on elements of the form 
$v_1\otimes\cdots \otimes v_N$, $\FF^N_j$ decomposes on elements of the form $(\prod\limits_{k=j+1}^N\tr{v_k})v_1\otimes\dots \otimes v_j$. Since one knows that $\FF^N_j$ is symmetric, 
it is enough to choose one bijection $i_J:\{1,\dots,j\}\to J,\ |J|=j$, and consider the mapping  
%case $\tr{(v_j)}=1$ and one associates to $\FF^N_j$ $\FF^N_J$ for any $L$ such that $|J|=j$ by 
\bea
\Phi_{i_J}:
\lone{|J|}&\overset{\Phi_{i_J}}\to&\lone{J}\nn\\
{\underset{j\in J}{\otimes}}v_j\in\lono{|J|}
&\mapsto&
\overset{N}{\underset{i=1}{\otimes}}a_i\in\lono{J} 
\label{ntoJloc}\\
 \FF^N_{|J|}&\mapsto&\FF^N_{J}\label{ntoJ}
 \eea
where  $a_s=1$ if $i\notin J$ and $a_{i_J(j)}=v_j$.
%are the ones of \eqref{margj} with $\tr{(v_j)}=1$.

\noindent $\Phi_{i_J}$ is obviously one-to-one since $i_J$ is so, and, though \eqref{ntoJloc} depends on the embedding chosen, \eqref{ntoJ} does not: $\Phi_{i_J}$ restricted to the space $\lone{|J|}_S$ of symmetric-by-permutation elements of $\lone{|J|}$, depends only on $J$ and not on $i_J$. We will call $\Phi_J$ this restriction,
\be\label{phionls}
\Phi_J=\Phi_{i_J}|_{\lone{|J|}_S}.
\ee
 
 The same argument is also valid for $E_J$ which enjoys the same symmetry property than $F^N_J$ and we define
% : choosing any embedding $i:\{1,\dots,j\}\to J.\ |J|=j$, the same mapping \eqref{ntoJloc} induces the mapping
% \be\label{ntoJE}
%  E_{\{1,\dots,|J|\}}\mapsto E_{J}
%  \ee which doesn't depend on the embedding chosen. 
%  
%  \noindent For $j\in\{1,\dots,N\}$, we will define $E_j\in\lone{j}$ by
%  \be\label{defej}
%  E_{\{1,\dots,j\}}=E_j\otimes 1_{\lone{N-j}}.
%  \ee
\be\label{defejjj}
E_{|J|}=\Phi_J^{-1}E_J.
\ee
  $\Phi_J$ is obviously isometric and we have that 
  \be\label{normej}
  \norm{E_J}_{\lone{J}}=\norm{E_{\{1,\dots,|J|\}}}_{\lone{\{1,\dots,|J|\}}}=\norm{E_{|J|}}_
 % {\lone{|J|}}
  1.
  \ee
Therefore, considering the one-to-one correspondence $\Phi_J$, it is enough to compute/estimate the quantities $E_j, j=1,\dots,N$. $E_j$ and $\FF^N_j$ are linked by
\be\label{defejj}
\left\{
\begin{array}{ccl}
E_j&=&\suml_{K \subset J} (-1)^{|K|} [\FF]^{\otimes K}_J\Phi_{J/K}\FF^N_{j-|K|}\\
&&\\
\FF^N_j&=&\suml_{K \subset J} [\FF]^{\otimes K}_J\Phi_{J/K}E_{j-|K|}\;.
\end{array}\right.
\ee
\begin{enumerate}[label=]
\item {\it \ksq the corrsponding expression are given in Sections  \ref{kac} and \ref{quantum}.}
\end{enumerate}
\subsection{Main results similar to \cite{PPS}}\label{resultsPPS}
% The proofs of Theorem 2.4,  Theorem 2.1 and Corollary 2.2 in \cite{PPS} leads to the following facts.
 
 The kinetic errors $ E_j,\ j=\unN,$
 satisfy the system of equations
\bea\label{eqhieraerror1}
\pa_t E_j&=& \left(K^j+\frac 1{N}T_j\right) E_j +
%E_{J}.
D_j
E_{{j}} \nn \\
&+&  
D_j^1
E_{j+1} + 
D_j^{-1}
E_{j-1} +
D_j^{-2}
E_{j-2},
\eea
where the  operators $D_j,D_j^1,D_j^{-1},D_j^{-2}, j=\unN$, are defined in Appendix \ref{proofmaineq} below, equations \eqref{newds}-\eqref{E0}, together with the proof of \eqref{eqhieraerror1}. Moreover, since \eqref{endow} holds true, we know by Remark 3.2 in \cite{PPS}, that the proof of Theorem 2.1 (and therefore Corollary 2.2) in \cite{PPS} remain valid in our present setting. 

We get the following result.
\begin{Prop}
The statements of Theorem \ref{mainpps} hold true in the abstract setting defined in Section \ref{model}.
\end{Prop}

\vskip 1.5cm
\subsection{Asymptotic expansion
 %and main result of the present article
 }\label{mainresultabs}\ 
It is easy to see that the proofs of the main results  expressed in Section \ref{mainresult} are adaptable in an elementary way to the present abstract paradigm. Indeed they use only the three properties stated in Remark \ref{crucial}, valid in the present setting as pointed out at the very end of Appendix \ref{proofmaineq}, formula \eqref{endow}, together with \eqref{K+}-\eqref{KV+}.

Therefore,  the statements contained in Theorem \ref{main} and Corollary  \ref{corcormain} hold true, verbatim, under the hypothesis of Theorem \ref{mainpps}, and with the definition of corrections errors given by the first line of \eqref{defejj} and replacing $\tfrac{{\norm{V}_{L^\infty}}}\hbar$ by $\norm{V}$ in \eqref{csas}. 

Moreover defining now $F^{N,n}_j$ by truncating the second line of \eqref{defejj} at order $n$, that is
$$
\FF^{N,n}_j=\suml_{K \subset J} [\FF]^{\otimes K}_J\Phi_{J/K}E^n_{j-|K|}
$$
where $E^n_j$ is defined by \eqref{Enj},
% in both cases,  by \eqref{FNnj}, the analog of
 Corollary \ref{cormain} reads as follows.
 
 \begin{Cor}\label{cormaina}[abstract]
Let $\FF^N(t)$ the solution of the $N$ body system \eqref{Neq}   with initial datum $\FF^N(0)=\FF^{\otimes N}$, $0<\FF\in \lon, \norm{F}_1=1$, and $\FF(t)$ the solution of the mean-field  equation \eqref{mfea}  with initial datum $\FF$.

Then,  for all $n\geq0$ and $N\geq 4(eA^{2n}_tj)^2$,
\bea
%&&\norm{\FF^N_j(t)-\FF(t)^{\otimes j}-N^{-1}\FF(t)^{\otimes (j-1)}\cE_{j,n}
%-\sum_{\ell=2}^jN^{-\ell/2}\FF(t)^{\otimes (j-\ell)}\cE_{j,n-\ell}}\nn\\
&&
%\norm{\FF^N_j(t)-\sum_{\ell=2}^j\FF(t)^{\otimes (j-\ell)}E_{j-\ell}^{n-\ell}}
\norm{\FF^N_j(t)-
%\suml_{K\subset \{\un{j}\}}\margetracej{K}{\FF}{\{\un{j}\}}\Phi_{\{\un{j}\}/K}E_{j-|K|}^{n
%%-\frac{(j-|K|)+((-1)^{j-|K|}-1)/2}2
%}}
%\sum_{K \subset \{\un{j}\}} \Pii{K}{j}\ \big(F(t)^{\otimes |K|}\otimes E^n_{j-|K|}(t)\big)
F^{N,n}_j(t)
}_1
\leq
%2C_{2n}(t)
N^{-n-\frac12}\ 
\tfrac{2tC_{2n}(t)eA^{2n}_tj}{\sqrt N}.\nn
\eea
%for $N\geq 4(eA^{2n}_tj)^2$.
\end{Cor}
%\section{Derivation of the abstract BBGKY hierarchy}\label{absbbgky}
%In this Appendix, we derive \eqref{eqhiera} and \eqref{absbbgkyhier} in the paradigm of Appendix \ref{abstract}, using \eqref{trackf} that we recall
%\be\label{trackfappend}
%\tr(KF)=0\mbox{ and }
%\tr^{j,N} (V_{i,r} F^N)=0,\ 1\leq j<i,r\leq N\;.
%\ee
%Thanks to \eqref{trackfappend} we get easily by applying $\tr^{j,N}$ on \eqref{Neq} that
%$$
%\frac d{dt}F^N_j=(K^j+\frac{T_j}N)F^N_j
%+\frac1N\tr^{j,N}(\sumls{1\leq i\leq j< k\leq N} V_{i,k}F^N)
%$$
%By symmetry of $F^N$ and $V_{i,k}$ we get $\tr^{j,N}(V_{i,k}F^N)=\tr^{j+1}(V_{i,k}F^N_{j+1})$ for $k>j$ and the result follows.
\section{Derivation of the correlation hierarchy \eqref{eqhieraerror1}}\label{proofmaineq}
%In this appendix we prove Proposition \ref{maineq}.

From the definition of $E_j$ (cf.\,\eqref{deferror}) we find
\bea
\pa_t E_J=
&\suml_{K \subset J} (-1)^{|K|} 
\left(\pa_t(\margetracej{K}{\FF}{J})\FF^N_{J/K}+
\margetracej{K}{\FF}{J}\pa_t\FF^N_{J/K}\right)\nonumber
\eea
Moreover, by \eqref{mgtr}
\be\label{kzero}
\pa_t\left(\margetracej{K}{\FF}{J}\right)=\sum_{k_0\in K}\margetracej{K/\{k_0\}}{\FF}{J}
\margetracej{\{k_0\}}{\pa_t\FF}{J/(K/\{k_0\})}.
\ee
% and since $\FF^N$ is symmetric by permutation of the tensorial factors, we get that, for any subset $J\in\{\unN\}$
%We recall the BBGKY hierarchy derived in Appendix \ref{absbbgky}, i.e.
Applying $\Phi_J$ defined in \eqref{defejjj} to the BBGKY hierarchy \eqref{eqhiera}, one finds easily that $F^N_J$ satisfies, denoting $\a(j,N):=\frac{N-j}N$,
\be\label{absbbgkyhier}
\pa_t \FF^N_J = K^J\FF^N_{J}+\frac1N\suml_{i<r\in J}T_{i,r}\FF^N_J+ \a(j,N) \suml_{i\in J}C_{i,j+1} \FF^N_{J\cup\{j+1\}}
\ee
(for $j+1\notin J$).
%where we wrote $K_J = \sum_{j \in J}K^j$.
%Finally, we recall 

By the mean-field equation \eqref{mfea} 
%$$
%\pa_t\FF=K\FF+Q(\FF,\FF)\;.
%$$
%From this 
we deduce that 
\bea
\pa_t E_J
=& \suml_{K \subset J} (-1)^{|K|} \suml_{k_0\in K} 
\margetracej{K/\{k_0\}}{\FF}{J}
 \margtracej{k_0}{(K\FF+Q(\FF,\FF))}{J/(K/\{k_0\}}\FF^N_{J/K} \nn \\
+&   \suml_{K \subset J} (-1)^{|K|}\a(j-|K|,N) \suml_{i\in J/K } \margetracej{K}{\FF}{J} C_{i,j+1} \FF^N_{(J/K)\cup \{j+1\}  }  \nn\\
+& \frac 1{2N} \suml_{K \subset J} (-1)^{|K|} \margetracej{K}{\FF}{J} (\suml_{ i\neq r \in J/K } T_{i,r})\FF^N_{J/K }\nn\\
+& \suml_{K \subset J} (-1)^{|K|} \margetracej{K}{\FF}{J}( K^{J/K }\FF^N_{J/K })\;.
\label{eq1} 
\eea
We denote by ${\cal T} _i$, $i=1,2,3,4$, the four terms contained  in the four lines of the r.h.s.\,of \eqref{eq1},
%We compute now the four terms in the r.h.s.\,of \eqref{eq1}, denoted by ${\cal T} _i$, $i=1,2,3,4$
respectively.
%Let us note first that $\forall k_0\in K,\ J/K=(J/\{k_0\})/(K/\{k_0\})$.
%As we shall see, t
The computation of the ${\cal T} _i$s is purely algebraic and will use only 
%the decompositions $\suml_{K\subset L}(-1)^{|K|}=\delta_{L,\emptyset},\ 
%\suml_{K\subset L}|K|(-1)^{|K|}=-\delta_{|L|,1}$ (cf.\,\eqref{defkro})
%and the 
%%properties \eqref{KK'}, \eqref{intercf}, which we rewrite for convenience:
%two 
the four following properties
$$
\left\{
\begin{array}{l}
 \suml_{K\subset L}(-1)^{|K|}=\delta_{|L|,\emptyset}\\ 
\suml_{K\subset L}|K|(-1)^{|K|}=-\delta_{|L|,1}\\
\margetracej{K}{\FF}{J}\margetracej{K'}{\FF}{J/K}
=\margetracej{K'}{\FF}{J}\margetracej{K'}{\FF}{J/K'}
=\margetracej{(K\cup K')}{\FF}{J},\ \ \ K,K'\subset J,\ K\cap K'=\emptyset\\
C_{i,j+1}\margetracej{K}{\FF}{(J/K)\cup \{j+1\}}=
\margetracej{K}{\FF}{(J/K)}C_{i,j+1},\ K\subset J,\ j+1\notin J.
\end{array}
\right.
$$
In order not to make the paper too heavy, we will compute extensively two terms and refer to an earlier version, \cite{psextenso}, for an exhaustive calculus.
\vskip 0.5cm
Using  the definition \eqref{deferror}, we get 
\bea
%&&\nn\\
&\cal T_1:=&\suml_{K \subset J} (-1)^{|K|} \suml_{k_0\in K}  
%\prodl_{ k\in K/\{k_0\}}
\margetracej{K/\{k_0\}}{\FF}{J}\margtracej{k_0}{(K\FF+Q(\FF,\FF))}{J/(K/\{k_0\})}\FF^N_{J/K} \nn \\
%&=&
%\suml_{k_0\in J}\suml_{K/\{k_o\}\subset J/\{k_o\}}
%(-1)^{|K/
%\{k_o\}|+1}\margtrace{k_0}{(K\FF+Q(\FF,\FF))}\nn\\
%&&\prodl_{k\in K/\{k_0\}}\margtrace{k}{\FF}\FF^N_{(J/\{k_0\})/(K/\{k_0\})}\nn\\
&=&
-\suml_{k_0\in J}\margtracej{k_0}{(K\FF+Q(\FF,\FF))}{J}
\suml_{K \subset J/\{k_0\}}(-1)^{|K|}
%\prodl_{ k\in K}
\margetracej{K}{\FF}{J/\{k_0\}}\FF^N_{(J/\{k_0\})/K} \nn\\
&=&
-\suml_{i\in J}\margtracej{i}{(K\FF+Q(\FF,\FF))}{J}E_{J/\{i\}}\;.\label{t1}
%&=&
%-\suml_{k_0\in J}\margtracej{k_0}{(K\FF+Q(\FF,\FF))}{J}
%\suml_{K \subset J/\{k_0\}}(-1)^{|K|}
%%\prodl_{ k\in K}
%\margetracej{K/\{k_0\}}{\FF}{J/\{k_0\}}
%\suml_{K'\subset J/K}
%%\prodl_{k'\in K'}
%\margetracej{K'}{\FF}{J/K}E_{(J/K)/K'}\nn\\
%&=&
%-\suml_{k_0\in J}\margtracej{k_0}{(K\FF+Q(\FF,\FF))}{J}
%\suml_{K \subset J/\{k_0\}}
%\suml_{K'\subset J/K}
%(-1)^{|K|}
%%\prodl_{ k\in K}
%\margetracej{K/\{k_0\}\cup K'}{\FF}{J/\{k_0\}}
%%\suml_{K'\subset J/K}
%%\prodl_{k'\in K'}
%%\margetracej{K'}{\FF}{J/K}
%E_{(J/K)/K'}\nn\\
%&=&
%-\suml_{k_0\in J}\margtracej{k_0}{(K\FF+Q(\FF,\FF))}{J}
%\sumls{L \subset J\\ K\subset L\\ k_0\notin K}(-1)^{|K|}
%%\prodl_{ l\in L}
%\margetracej{L}{\FF}{J/\{k_0\}}
%E_{(J/\{k_0\})/L}\nn
\eea
%But $$\suml_{L \subset J}\suml_{\substack{K\subset L\\k_0\notin K}}(-1)^{|K|}\dots\ =
%\suml_{\substack{L \subset J\\k_0\notin L}}(\suml_{\substack{K\subset L}}(-1)^{|K|})\dots\ +
%\sumls{L \subset J\\k_0\in L}
%(\suml_{\substack{K\subset L/\{k_0\}}}(-1)^{|K|})\dots$$
% Since 
%$\suml_{\substack{K\subset L}}(-1)^{|K|}=\delta_{L,\emptyset}$ by \eqref{defkro}, we get
%\bea
%\mathcal T_1&=&
%-\suml_{k_0\in J}\margtrace{k_0}{(K\FF+Q(\FF,\FF))}E_{J}
%\eea

%We introduce the notation: $K^i= K/ \{ i \}$, therefore ($ K \to K=K^i$ )
%\bea
%{\cal T} _1=& \sum_{i\in J} \sum_{K \subset J} (-1)^{(|K|-1)}  f^{\otimes K} Q(f,f)_i f^N_{J^i/K}  \nn \\
%&= - \sum_{i\in J} Q(f,f)_i E_{J^i} .
%\eea
%Here we used definition \eqref{deferror}.

To compute ${\cal T} _2$ we make use of the inverse definition \eqref{invdeferror}:
\bea\label{calt2}
{\cal T} _2&:=&
%\a(j,N) \suml_{K \subset J} (-1)^{|K|} \suml_{i\in J } \prodl_{k\in K}\margtrace{k}{\FF} C_{i,j+1} 
%\suml_{K'\subset J\cup\{j+1\}}\prodl_{k'\in K'}\margtrace{k'}{\FF}E_{J\cup\{j+1\}}
\suml_{K \subset J}\a(j-|K|,N) (-1)^{|K|} \suml_{i\in J/K } \margetracej{K}{\FF}{J} C_{i,j+1} \FF^N_{(J/K)\cup \{j+1\}  } \nn\\
&=&
 \suml_{K \subset J}\a(j-|K|,N) (-1)^{|K|} \suml_{i\in J/K } \margetracej{K}{\FF}{J}\dots \nn\\
&&
\dots C_{i,j+1}\sum_{K'\subset (J/K)\cup \{j+1\}}
\margetracej{K'}{\FF}{(J/K)\cup \{j+1\}}E_{((J/K)\cup \{j+1\})/K'}\;.
\eea
Distinguishing among the belonging or not to $K'$ of $i$ and $j+1$ in the r.h.s. of \eqref{calt2},   we decompose 
\be\label{deft2}
\mathcal T_2=
%\suml_{s=1}^4\mathcal T_2^s
\mathcal T_2^{i, j+1 \in K'}
+
\mathcal T_2^{i, j+1 \notin K'}
+
\mathcal T_2^{i \in K' , j+1 \notin  K'}
+
\mathcal T_2^{i \notin K' , j+1 \in K'}
\ee 
%where each $\mathcal T_2^s$ concerns the part of the sum corresponding to the following constraints:
%\begin {itemize}
%\item  $s=1:$ \quad  $i, j+1 \in K'$\;,
%\item   $s=2:$  \quad $i, j+1 \notin K'$\;,
%\item   $s=3:$  \quad $i \in K' , j+1 \notin  K'$\;,
%\item   $s=4:$  \quad $i \notin K' , j+1 \in K'$\;.
%\end{itemize}
%%\newpage
%For $s=1$ w
We have
\bea
{\cal T}^{i, j+1 \in K'} _2
&=&
\suml_{K \subset J}\a(j-|K|,N) (-1)^{|K|} \suml_{i\in J/K } \margetracej{K}{\FF}{J} \dots\nn\\
&&
\dots C_{i,j+1}\sumls{K'\subset (J/K)\cup \{j+1\}\\i,j+1\in K'}
\margetracej{K'}{\FF}{(J/K)\cup \{j+1\}}E_{((J/K)\cup \{j+1\})/K'}\nn\\
&=&
\suml_{K \subset J}\a(j-|K|,N) (-1)^{|K|} \suml_{i\in J/K } \margetracej{K}{\FF}{J} \dots\nn\\
&&
\dots C_{i,j+1}\sumls{K''\subset (J/K)/\{i\}}
\margetracej{K''\cup\{i,j+1\}}{\FF}{(J/K)\cup \{j+1\}}E_{(J/K)/(K''\cup\{i\})}\nn\\
&=&
\suml_{K \subset J/\{i\}} \a(j-|K|,N)(-1)^{|K|} \suml_{i\in J/K } \margetracej{K}{\FF}{J} \dots\nn\\
&&
\dots C_{i,j+1}\sumls{K''\subset (J/\{i\})/K}
\margetracej{K''\cup\{i,j+1\}}{\FF}{(J/K)\cup \{j+1\}}E_{(J/K)/(K''\cup\{i\})}\nn\\
&=&
 \suml_{K \subset J/\{i\}}\a(j-|K|,N) (-1)^{|K|} \suml_{i\in J } \margetracej{K}{\FF}{J} \dots\nn\\
&&
\dots\sumls{K''\subset (J/\{i\})/K}
\margetracej{K''}{\FF}{(J/K)}
 C_{i,j+1}\margetracej{\{i,j+1\}}{\FF}{((J/K)/K'')\cup \{j+1\}}
 E_{(J/K)/(K''\cup\{i\})}\nn\\
 &=&
 \suml_{i\in J }\suml_{K \subset J/\{i\}}\a(j-|K|,N) (-1)^{|K|}  \margetracej{K}{\FF}{J} \sumls{K''\subset (J/\{i\})/K}
\margetracej{K''}{\FF}{(J/K)}\dots\nn\\
&&
\dots
 C_{i,j+1}\margetracej{\{i,j+1\}}{\FF}{((J/K)/K'')\cup \{j+1\}}
 E_{(J/K)/(K''\cup\{i\})}\nn\\
 &=&
 \suml_{i\in J }\suml_{L \subset J/\{i\}}(\sumls{K\subset L} \a(j-|K|,N)(-1)^{|K|})  \margetracej{L}{\FF}{J} 
\dots\nn\\
&&
\dots
 C_{i,j+1}\margetracej{\{i,j+1\}}{\FF}{((J/L)\cup \{j+1\}}
 E_{J/(L\cup\{i\})}\nn
 %\\
 \eea
 \bea
 &=&
\a(j,N) \suml_{i\in J }
 C_{i,j+1}\margetracej{\{i,j+1\}}{\FF}{J\cup \{j+1\}}
 E_{J/\{i\}}\nn\\
 &&
 -\frac1N
 \suml_{i\neq l\in J }
 %\suml_{l\in J/\{i\}} 
 \margetracej{\{l\}}{\FF}{J} 
%\dots\nn\\
%&&
%\dots
 C_{i,j+1}\margetracej{\{i,j+1\}}{\FF}{(J/\{l\})\cup \{j+1\}}
 E_{J/(\{i,l\})}\nn\\
 &=&
\a(j,N) \suml_{i\in J }
 \margetracej{\{i\}}{Q(\FF,\FF)}{J}
 E_{J/\{i\}}\nn\\
 &&
 -\frac1N
 \suml_{i\neq l\in J }
 %\suml_{l\in J/\{i\}}  
%\dots\nn\\
%&&
%\dots
 C_{i,j+1}\margetracej{\{l\}}{\FF}{J\cup \{j+1\}}\margetracej{\{i,j+1\}}{\FF}{(J/\{l\})\cup \{j+1\}}
 E_{J/(\{i,l\})} 
 \nn\\
 &=&
\a(j,N) \suml_{i\in J }
 \margetracej{\{i\}}{Q(\FF,\FF)}{J}
 E_{J/\{i\}}
 -\frac1N
 \suml_{i\neq l\in J }
 %\suml_{l\in J/\{i\}}  
%\dots\nn\\
%&&
%\dots
 C_{i,j+1}\margetracej{\{i,l,j+1\}}{\FF}{J\cup \{j+1\}}
 E_{J/(\{i,l\})} \nn
\eea
since 
$\suml_{\substack{K\subset L}}(-1)^{|K|}=\delta_{L,\emptyset}$. 
%We get
%\bea
%{\cal T}^1 _2
%&=&
%\a(j,N) \suml_{i\in J } 
%\margetracej{\{i\}}{Q(\FF,\FF)}{J}E_{J/\{i\}}\nn\\
%&-&\frac1N
%\suml_{i\neq l\in J } 
%\margetracej{\{i\}}{Q(\FF,\FF)}{J}
%\margetracej{\{l\}}{F}{J/\{i\}}
%E_{J/\{i,l\}}
%\label{t12}
%\eea
%and therefore 
Note that there is a crucial compensation:
\bea\label{t1t12}
{\cal T}_1+{\cal T}^{i, j+1 \in K'} _2&=&-\frac jN \suml_{i\in J } 
\margetracej{\{i\}}{Q(\FF,\FF)}{J}E_{J/\{i\}}
\nn\\
&-&\frac1N
\suml_{i\neq l\in J } 
\margetracej{\{i\}}{Q(\FF,\FF)}{J}
\margetracej{\{l\}}{F}{J/\{i\}}
E_{J/\{i,l\}}.
\eea

 The computations of $\mathcal T_2^{i, j+1 \notin K'},\ \mathcal T_2^{i \in K' ,\ j+1 \notin  K'}
.\ \mathcal T_2^{i \notin K' , j+1 \in K'}$ go the same way and we omit it here (see \cite{psextenso} for a complete derivation).

We consider a similar dichotomy for the term
\bea
{\cal T} _3&:=&\frac 1{2N} \suml_{K \subset J} (-1)^{|K|} \margetracej{K}{\FF}{J} (\suml_{ i\neq r \in J/K } T_{i,r})\FF^N_{J/K }\nn\\
&=&\frac 1{2N} \suml_{K \subset J} (-1)^{|K|} \margetracej{K}{\FF}{J} (\suml_{ i\neq r \in J/K } T_{i,r})
\suml_{K' \subset J/K}  \margetracej{K'}{\FF}{J/K}E_{J/(K\cup K')}\;.
\nn
\eea
%We set
%\be
%{\cal T}_3=\sum_{s=1}^4 {\cal T}^i _3
%\ee
%where each term ${\cal T}^i _3$ corresponds to the constraints
according, this time, to the cases 
  $i, r \in K'$,
 $i, r \notin K'$,
 $i \in K' , r \notin  K'$
and $i \notin K' , r \in  K'$. The computation of the different terms uses the same ``tricks" than for $\cal T_2$ and we omit them (see again \cite{psextenso}).

Finally, we obtain easily (see \cite{psextenso}) that
\bea\label{t4}
\mathcal T_4
&:=&
 \suml_{K \subset J} (-1)^{|K|} 
% \prodl_{k\in K}\margtrace{k}{\FF} K_{J }\FF^N_{J }
 \margetracej{K}{\FF}{J}(K^{J/K}\FF^N_{J/K})=
%\nn\\
%&=&
%\suml_{K \subset J}\ 
%\suml_{K' \subset J/K}(-1)^{|K|} 
%\margetracej{K}{\FF}{J}K^{J/K}\margetracej{K'}{\FF}{J/K}E_{(J/K)/K'}
%%  \suml_{j\in J}\suml_{K \subset J} (-1)^{|K|}\prodl_{k\in K}\margtrace{k}{\FF} K^j
%%\suml_{K'\subset J}\prodl_{k'\in K'}\margtrace{k'}{\FF}E_J
%\nn\\
%%&=&
%% \suml_{j\in J}\suml_{K \subset J} (-1)^{|K|} \prodl_{j\neq k\in K}\margtrace{k}{\FF} 
%% \margtrace{j}{\KK}K^j
%%\suml_{K'\subset J}\prodl_{k'\in K'}\margtrace{k'}{\FF}E_J\nn
%&=&
%\suml_{K \subset J}
%\suml_{K' \subset J/K}(-1)^{|K|} 
%\margetracej{K}{\FF}{J}(\sum_{j\in J/K}K^j\margetracej{K'}{\FF}{J/K}E_{(J/K)/K'})\ \nn\\
%&=&
%\suml_{K \subset J}
%\sum_{j\in J/K}K^j
%\suml_{K' \subset J/K}(-1)^{|K|} 
%\margetracej{K}{\FF}{J}\margetracej{K'}{\FF}{J/K}E_{(J/K)/K'}\nn\\
%&=&
%\sum_{j\in J}K^j
%\suml_{K \subset J/\{j\}}
%\suml_{K' \subset J/K}(-1)^{|K|} 
%\margetracej{K}{\FF}{J}\margetracej{K'}{\FF}{J/K}E_{(J/K)/K'}\nn\\
%&=&
%\sum_{j\in J}K^j
%\suml_{L \subset J/\{j\}}
%(\suml_{K \subset L}
%(-1)^{|K|} )
%\margetracej{L}{\FF}{J}E_{(J/L)}\nn\\
%&=&
%\sum_{j\in J}K^j
%E_{(J)}\nn\\
%&=&
K^J
E_{J}.
\eea
Summing up all the contributions $\cal T_1,1=1,\dots,4$ (see \cite{psextenso}), we get \eqref{eqhieraerror1} after specializing 
%\eqref{eqhieraerrorJ} 
to the case $J=\{1,\dots,j\}$, using \eqref{defejjj} and setting

\bea
D_j&:& \ \lone{j}\to\lone{j},\ j=\unN,\nn\\
&&E_j\mapsto\frac{N-j}N \suml_{ i\in J}  C_{i,j+1} \left(\margetracej{\{i\}}{\FF}{J\cup\{j+1\}}  \Phi_{{(J\cup\{j+1\})/\{i\}}}E_j
+ \margetracej{\{j+1\}}{\FF}{J\cup\{j+1\}}E_j\right)\;,\nn\\
&&\hskip 1.2cm
-\frac1N
 \suml_{i\neq l\in J }
 %\margetracej{\{l\}}{\FF}{J} 
 C_{i,j+1}\big( \margetracej{\{l\}}{\FF}{J\cup\{j+1\}}\Phi_{(J/\{l\})\cup \{j+1\}}E_{j}\big)\nn \\
D^1_j&:& \lone{(j+ 1)}\to\lone{j},\ j=1,\dots,N-1,\nn\\
&& E_{j+1}\mapsto
\frac{N-j}N C_{j+1} 
%\Phi_{J\cup\{j+1\}}
E_{j+1}\;,\nn\\ 
D^{-1}_j&:& \lone{(j- 1)}\to\lone{j}\ j=2,\dots,N,\nn\\
&& E_{j-1}\mapsto
\left(-\frac jN  \suml_{i\in J} \margetracej{\{i\}}{Q(\FF,\FF)}{J} + \frac 1{2N} \suml_{i,r \in J}  T_{i,r} \margetracej{\{i\}}{\FF}{J}\right) \Phi_{J/\{i\}}E_{j-1}\;, \nn\\ 
&&\hskip 1.5cm-
\frac1N
\suml_{i\neq l\in J }
\margetracej{\{l\}}{\FF}{J} 
C_{i,j+1} \margetracej{\{j+1\}}{\FF}{(J/\{l\})\cup\{j+1\}} 
\Phi_{J/\{l\}}E_{j-1} \nn\\
&&\hskip 1.5cm-
\frac1N
\suml_{i\neq l\in J }
\margetracej{\{l\}}{\FF}{J}
C_{i,j+1}\margetracej{\{i\}}{\FF}{(J/\{l\})\cup\{j+1\}}\Phi_{(J/\{i,l\})\cup \{j+1\}}E_{j-1}\nn \\
D^{-2}_j&:& \ \lone{(j-2)}\to\lone{j},\ j=3,\dots,N,\nn\\
&&E_{j-2}\mapsto
\frac 1{2N}  \suml_{i,s \in J}  T_{i,s} \margetracej{\{i\}}{\FF}{J} \margetracej{\{s\}}{\FF}{J/\{i\}} \Phi_{J/\{i,s\}}E_{j-2}\;\nn\\
&&\hskip 1.6cm 
-\frac1N
\suml_{i\neq l\in J } 
\margetracej{\{i\}}{Q(\FF,\FF)}{J}
\margetracej{\{l\}}{F}{J/\{i\}}
\Phi_{J/\{i,l\}}E_{j-2}.\label{newds}
\eea
where, by convention, 
%\bea
%\label{E0}
%&& D_N^1=D_1^{-2}=0\nn\\
%&& D_1^{-1} \left(E_0\right):=-\frac 1N Q(F,F)\;,
%% \quad D_1^{-2}\left(E_{-1}\right):=0\;,
%\\
%&& D_2^{-2}\left(E_0\right):= \frac 1 N  T_{1,2} F \otimes F - \frac{1}{N} Q(F,F)^{\otimes\{1\}} 
%F^{\otimes \{2\}} -  \frac{1}{N} Q(F,F)^{\otimes\{2\}} 
%F^{\otimes \{1\}} \;.\nn
%\eea
\be
\label{E0}
\left\{
\begin{array}{l}
 D_N^1:=D_1^{-2}:=0\\
 D_1^{-1} \left(E_0\right):=-\frac 1N Q(F,F)\;,
% \quad D_1^{-2}\left(E_{-1}\right):=0\;,
\\
D_2^{-2}\left(E_0\right):= \frac 1 N  \big(T_{1,2} (F \otimes F) -  Q(F,F)\otimes 
F -  F\otimes Q(F,F) 
\big) \;.
\end{array}
\right.
\ee
Note that one has the following estimates:
\be\label{endow}
\| D_j \| , \| D^1_j\| \leq j  \mbox{ and }  \| D_j^{-1}\|,
\| D_j^{-2}\|, \|D^{-1}_1(E_0)\|,  \|D^{-2}_2(E_0)\|
 \leq \frac {j^2} N.
\ee

}
\end{appendix}
%\newpage
%\begin{thebibliography}{99}
%\bibitem{BPS}
%N. Benedikter‚àó, M. Porta‚Ä  and B. Schlein: 
%\textit{Effective Evolution Equations from Quantum Dynamics}, 
%preprint February 10, 2015.
%\bibitem{LNS} M. Lewin, P. T. Nam and B. Schlein, \textit{Fluctuations around Hartree states in the mean-field
%regime} Preprint arXiv:1307.0665.
%\bibitem{PPS} T. Paul, M. Pulvirenti, S. Somonella, \textit{On the size of kinetic chaos for mean filed models},
%preprint.
%\end{thebibliography}

\end{document}